\documentclass[10pt,letter]{article}
\usepackage[papersize={8.5in,11in}]{geometry}

\usepackage{algorithm}
\usepackage{algorithmic}
\usepackage{amssymb}
\usepackage{amsmath}
\usepackage{amsthm} 
\usepackage{amsopn}
\usepackage{graphics}
\usepackage{epsfig}
\usepackage{color}
\usepackage{enumerate}
\usepackage{pdfsync}

\usepackage{hyperref}
\newcommand{\namedref}[2]{\hyperref[#2]{#1~\ref*{#2}}}
\newcommand{\equationref}[1]{\namedref{Equation}{#1}}
\newcommand{\sectionref}[1]{\namedref{Section}{#1}}
\newcommand{\appendixref}[1]{\namedref{Appendix}{#1}}
\newcommand{\theoremref}[1]{\namedref{Theorem}{#1}}
\newcommand{\lemmaref}[1]{\namedref{Lemma}{#1}}
\newcommand{\definitionref}[1]{\namedref{Def.}{#1}}
\newcommand{\algorithmref}[1]{\namedref{Algorithm}{#1}}
\newcommand{\factref}[1]{\namedref{Fact}{#1}}

\newcommand{\argmax}{\operatornamewithlimits{arg\,max}}

\newcommand{\prob}{\mathsf{Pr}}

\newtheorem{theorem}{Theorem}[section]
\newtheorem{lemma}[theorem]{Lemma}
\newtheorem{definition}[theorem]{Definition}

\newtheorem{claim}[theorem]{Claim}
\newtheorem{fact}[theorem]{Fact}
\newtheorem{remark}[theorem]{Remark}

\newcommand{\vp}{{\bf p}}
\newcommand{\vq}{{\bf q}}
\newcommand{\vd}{{\bf d}}
\newcommand{\ud}{\underline{\vd}}

\newcommand{\vv}{{\bf v}}

\newcommand{\inft}{\mathrm{T}}

\newcommand{\med}{\mathrm{med}}
\newcommand{\zero}{{\bf 0}}
\newcommand{\one}{{\bf 1}}
\newcommand{\uq}{\underline{\vq}}
\newcommand{\oq}{\overline{\vq}}
\newcommand{\iter}{{(\infty)}}
\newcommand{\vs}{{\bf s}}
\newcommand{\vx}{{\bf x}}
\newcommand{\vy}{{\bf y}}
\newcommand{\vr}{{\bf r}}
\newcommand{\vu}{{\bf u}}
\newcommand{\vl}{{\boldsymbol \ell}}

\begin{document}


\title
{Optimal Pricing in Social Networks \\ with Incomplete Information}

\newcommand{\email}[1]{\texttt{\color{black}#1}}

\newcommand*\samethanks[1][\value{footnote}]{\footnotemark[#1]}

\author{
Wei Chen\textsuperscript{1}\and
Pinyan Lu\textsuperscript{1}\and
Xiaorui Sun\textsuperscript{2,}\thanks{Part of this work was done while the authors were visiting Microsoft Research Asia.}\and
Bo Tang\textsuperscript{3,}\samethanks\and
Yajun Wang\textsuperscript{1}\and
Zeyuan~Allen~Zhu\textsuperscript{4,}\samethanks\ \textsuperscript{,}\thanks{A preliminary version of this work has appeared as a chapter of the B.Sci thesis of this author \cite{Zhu10}.} \\ \\
\vspace{-3pt}\small \textsuperscript{1} Microsoft Research Asia. \email{\{weic,pinyanl,yajunw\}@microsoft.com} \\
\vspace{-3pt}\small \textsuperscript{2} Columbia University. \email{xiaoruisun@cs.columbia.edu} \\
\vspace{-3pt}\small \textsuperscript{3} Shanghai Jiaotong University. \email{tangbo1@sjtu.edu.cn} \\
\vspace{-3pt}\small \textsuperscript{4} MIT CSAIL. \email{zeyuan@csail.mit.edu}
}


\maketitle

\begin{abstract}
In revenue maximization of selling a digital
product in a social network, the utility of an agent is often considered to have two parts:
a private valuation, and linearly additive influences from other agents.
We study the incomplete information case where agents know a common distribution about others' private valuations, and make decisions simultaneously.
The ``rational behavior'' of agents in this case is captured by the well-known Bayesian Nash equilibrium.

Two challenging questions arise: how to \emph{compute} an equilibrium and how to \emph{optimize} a pricing strategy accordingly to maximize the revenue assuming agents follow the equilibrium? In this paper, we mainly focus on the natural model where the private valuation of each agent is sampled from a uniform distribution, which turns out to be already challenging.

Our main result is a polynomial-time algorithm that can \emph{exactly}
compute the equilibrium and the optimal price, when pairwise influences are non-negative. If negative influences are allowed, computing any equilibrium even approximately is PPAD-hard.
Our algorithm can also be used to design an FPTAS for optimizing discriminative price profile.
\end{abstract}

\setcounter{tocdepth}{2}
\setcounter{page}{0}

\section{Introduction}
\label{sec:intro}
%
%
Social influence in large social networks provides huge monetization
        potential, which is under intensive
        investigation by companies as well as
        research communities.
Many digital products exhibit explicit social
values. For example, Zune players can share music with each other, so
the utility one can expect from a Zune player partially depends on
the number of her friends having the same product.
In a more direct case of instant
messaging, the utility for one user is critically determined by the
number of her friends who use the same instant messenger.
Therefore, how to design, market, and price products with external
        social values depends intimately on the understanding and
        utilization of social influence in social networks.





In this paper, we study the problem of selling a digital product to agents in a social network.
To incorporate social influence,
we assume each agent's utility of having
the product is the summation of two parts: the private intrinsic
valuation and the overall influence from her
friends who also have the product. In this paper, we study the linear influence case, i.e., the
overall influence is simply the summation of
influence values from her friends who have the product.

Given such assumption, the purchasing decision of one agent is not solely made
        based on her own
valuation, but
        also on information about her friends' purchasing decisions.
However, a typical agent does not have complete information about others' private valuations,
        and thus might make the decision based on her belief of
other agents' valuations.

In this paper, we study the case when this belief forms a
public
distribution, and rely on the solution concept of
        Bayesian Nash equilibrium~\cite{Harsanyi1967}.
Specifically, each agent knows her own private valuation (also referred to
        as her {\em type}); in addition, there is a distribution of this private valuation, publicly known
        by everyone in the network as well as the seller.
In this
        paper, we study the case that the joint distribution is a product
        distribution, and the valuations for all agents are sampled
        independently from possibly different uniform distributions.

\paragraph{Computing the Equilibria.}

Usually, there exist multiple equilibria in this game.
We first study the case when all influences are {\em non-negative}.
We show that there exist two special ones: the {\em pessimistic equilibrium}
and the {\em optimistic equilibrium}, and all other equilibria are between these two.
We then design a polynomial time algorithm to compute the pessimistic
(resp. optimistic) equilibrium \emph{exactly}.

The overall idea is to utilize the fact that the pessimistic
(resp. optimistic) equilibrium is ``monotonically increasing'' when
the price increases. However, the iterative method requires
exponential number of steps to converge, just like many potential
games which may well be PLS-hard.
Our algorithm is based on the line sweep paradigms, by increasing the
price $p$ and computing the equilibrium on the way. There are several
challenges we have to address to implement the line sweep
algorithm. See \sectionref{SUBSEC:OUTLINE}
for more discussions on the difficulties.
%

On the negative side, when there exist negative influences among
agents, the monotone property of the equilibria does not hold. In fact,
we show that computing an approximate equilibrium is PPAD-hard
for a given price, by a reduction from the two player Nash
equilibrium problem.


\paragraph{Optimal Pricing Strategy.}
When the seller considers
offering a uniform price,
our proposed line sweep algorithm
calculates the equilibrium as a function of the price.
This closed form allows us to find the price for the optimal revenue.

We also discuss the extensions to discriminative pricing
setting:
agents are partitioned into $k$
groups and the seller can offer different prices to different
groups. Depending on whether the algorithm can choose the partition or
not, we discuss the hardness and approximation algorithms of these
extensions.

%
%

\subsection{Related Work}
\label{app:related-work}
\paragraph{Influence maximization.}
Cabral et al.~\cite{Cabral1999} studied the property of the
optimal prices over time with network externality and strategic
agents. They show that the seller might set a low introductory price to attract a critical mass of
agents.
Another notable body of work in computer science is the \emph{optimal seeding} problem~(e.g. Kempe et al.~\cite{Kempe2003} and Chen et al.~\cite{Chen2009}), in which a set of $k$ seeds are selected to maximize the total influence according to some stochastic propagation model.

\paragraph{Pricing with equilibrium models.}
When there is social influence, a large stream of literature is focusing on simultaneous games. This is also known as the ``two-stage'' game where the seller sets the price in the first stage, and agents play a one-shot game in their purchasing decisions. Agents' rational behavior in this case is captured by the Nash equilibrium (or Bayesian Nash equilibrium if the information is incomplete).

The concept and existence of pessimistic and optimistic equilibria is
not new. For instance, in analogous problems with externalities,
Milgrom and Roberts \cite{Milgrom-Roberts1990} and Vives
\cite{Vives1990} have witnessed the existence of such equilibria
in the {\em complete information} setting. Notice that our pricing
problem, when
restricted
to complete information, can be trivially solved by an iterative method.

In incomplete information setting, Vives and Van Zandt
\cite{Vives-VanZandt2007} prove a similar existential result using
iterative methods.
However, they do not provide any convergence guarantee.
In our setting, we have shown in \sectionref{app:method} that such type of iterative methods may take exponential time to converge. Our proposed algorithm instead \emph{exactly} computes the equilibrium, through a much move involved (but constructive) method. In parallel to this work, Sundararajan \cite{Sundararajan2008} also discover the monotonicity of the equilibria, but for symmetry and limited knowledge of the structure (only the degree distribution is known).

It is worth noting that those works above do consider the case when
influence is not linear (but for instance supermodular). Though our paper focuses on
linear influences, our
monotonicity results for equilibria do easily extend to non-linear ones. See \sectionref{SEC:PRE}.

When the influence is linear, Candogan, Bimpikis and Ozdaglar
\cite{Candogan-Bimpikis-Ozdaglar-2010} study the problem with
(uniform) pricing model for a divisible good on sale.
It differs from our paper in the model: they are in complete information and divisible good setting;
more over, they have relied on a diagonal dominant assumption, which simplifies the problem and ensures the uniqueness of the equilibrium.

Another paper for linear influence is by Bloch and Querou
\cite{Bloch-Querou2009}, which also studies the uniform pricing
model. When the influence is small, they approximate the influence
matrix by taking the first 3 layers of influence, and then an
equilibrium can be easily computed. They also provide experiments to
show that the approximation is numerically good for random inputs.

\paragraph{Pricing with cascading models.}
In contrast to the simultaneous-move game considered by us (and many
others), another stream of work focuses on the cascading models with social influence.

Hartline, Mirrokni and Sundararajan~\cite{Hartline2008}  study the {\em explore and exploit} framework. In their model the seller offers the product to
the agents in a sequential manner, and assumes all agents are {\em
  myopic}, i.e., each agent is making the decision based
on the known results of the previous agents in the sequence. As they have pointed out, if the
pricing strategy of the seller and the private value distributions of
the subsequent agents are publicly known, the agents can make more
``informed'' decisions than the myopic ones. In contrast to them, we
consider ``perfect rational'' agents in the simultaneous-move game, where
agents make decisions {\em in anticipation} of what others may do given their
beliefs to other agents' valuations.

Arthur et al.~\cite{Arthur2009} also use the explore and exploit framework, and study a similar problem; potential buyers do not arrive sequentially as in~\cite{Hartline2008}, but can choose to buy the product with some probability only if being recommended by friends.

Recently, Akhlaghpour et al.~\cite{Akhlaghpour2010} consider the
multi-stage model
that the seller sets different prices for each stage. In contrast to
\cite{Hartline2008}, within each stage, agents are ``perfectly
rational'', which is characterized by the pessimistic equilibrium in
our setting with {\em complete information}.
As mentioned in~\cite{Akhlaghpour2010}, they did not consider
the case where a rational agent may defer her decision to later stages in
order to improve the utility.

\paragraph{Other works.}
If the value of the product does not exhibit social influence, the
seller can maximize the revenue following the optimal auction process by
the seminal work of Myerson~\cite{Myerson1981}. Truthful auction
mechanisms have also been studied for digital goods, where one can
achieve constant ratio of the profit with optimal fixed
price~\cite{Goldberg2006,Hartline2005}.
On computing equilibria for problems that guarantees to find an equilibrium through iterative methods, most of them, for instance the famous congestion game, is proved to be PLS-hard~\cite{PLScomplete}.

\section{Model and Solution Concept}
\label{SEC:PRE}

We consider the sale
of one digital product by a seller with zero cost, to the set of agents
$V=[n]=\{1,2,\ldots,n\}$ in a social network. The network is
 modeled as a simple {\em directed} graphs $G=(V,E)$ with no
self-loops.


\begin{itemize}
\item {\bf Valuation}: Agent $i$ has a private value $v_i\geq 0$
  for the product. We assume $v_i$ is sampled from a uniform
  distribution with interval $[a_i,b_i]$ for $0\leq a_i < b_i$, which we denote as
  $U(a_i,b_i)$.
  The values $a_i$ and $b_i$ are common knowledge.
\item {\bf Price}: We consider the seller offering the product at a
uniform
price $p$.
  We postpone discriminative pricing models in \appendixref{SEC:MULTIPRICE}.
\item {\bf Revenue:} Let $\vd=\{d_1,\ldots,d_n\} \in \{0,1\}^n$ be the
  decision vector the agents make, i.e., $d_i = 1$ if agent $i$ buys the
  product and $0$ otherwise. The revenue of the seller is defined as
  $\sum_i p\cdot d_i$.
When the decisions are random variables, the revenue is defined as
the expected payments received from the users.
\item {\bf Influence}: Let matrix
  $\inft = (\inft_{j,i})$ with $\inft_{j,i}\in \mathbb{R}$ and
  $i,j\in V$
  represent the influences among agents, with $\inft_{j,i} =0$ for
  all $(j,i)\notin E$. In particular, $\inft_{j,i}$ is the utility that agent $i$ receives
  from agent $j$, if both of them buy the product.
Except for the hardness result, we consider
$\inft_{j,i}$ to be
  non-negative.
\item {\bf Utility}:
Let $\vd_{-i}$ be the decision vector
  of the agents other than agent $i$. For convenience, we denote
  $\langle d_i',\vd_{-i} \rangle$ the vector by replacing the $i$-th entry of $\vd$ by
  $d_i'$.  In particular, given the influence matrix $\inft$, the
  utility is defined as:
  \begin{equation}
   \label{eqn:utility}
    u_i( \langle d_i,\vd_{-i} \rangle, v_i, p) = \left\{
    \begin{array}{ll}
v_i - p +\sum_{j\in [n]}d_j\cdot \inft_{j,i}, &\mbox{if } d_i = 1\\
0,&\mbox{if $d_i = 0$}
    \end{array}
  \right.
  \end{equation}
\end{itemize}

\begin{remark}
In the definition, we require $a_i< b_i$. This
condition can be relaxed to $a_i\leq b_i$, i.e., we are able to handle
the fixed value case as well. For instance, this only requires a
separate case analysis in our proposed line sweep algorithm
in \sectionref{SEC:BAYESIAN}. However, for
ease of
presentation,
we assume $a_i < b_i$ in the
remaining of the
paper, unless otherwise noted.
\end{remark}

Our question is: ``which price shall
the seller offer
to maximize the total revenue?''
In order to answer this question, it is
necessary to study the agents' {\em rational behavior} using the concept Bayesian Nash equilibrium (BNE).
For ease of presentation, we redefine the equilibrium based on the buying
probability of the agents. We will show that they are equivalent. Its proof
is in \appendixref{app:pre}.
\begin{definition}\label{def_probeq}
  The probability vector
  $\vq=(q_1,q_2,...,q_n)\in [0,1]^n$ is an {\em equilibrium at price $p$}, if (where $\med$ is the median function)
\begin{small}
\begin{align}
\forall i\in [n]\;,  q_i &=
  \prob_{v_i \sim U(a_i,b_i)}\Big[v_i - p + \sum_{j\in [n]} { T_{j,i} \cdot
    q_j} \geq 0\Big] \notag \\
& =\med{\left\{0,1,\frac{b_i-p+\sum_{j\in
        [n]}{T_{j,i}q_j}}{b_i-a_i}\right\}}. \label{eq_probeq0}
\end{align}
\end{small}
\end{definition}


\begin{lemma}
\label{lem:eq_eq}
  Given equilibrium $\vq$, the strategy profile such agent $i$ ``buys
  the product if and only if her internal valuation $v_i\geq p-\sum_{j\neq
    i}{T_{j,i}q_j}$'' is a Bayesian Nash equilibrium; on the contrary, if
  a strategy profile is a Bayesian Nash equilibrium, then the
  probability that agent $i$ buys the product satisfies \equationref{eq_probeq0}.
\end{lemma}


\equationref{eq_probeq0} can be also defined in the language of a transfer function, which we will extensively reply on in the rest of the paper.


\begin{definition}[Transfer function]
\label{def:transfer}
Given price $p$, we define the transfer function $f_p:[0,1]^n
\rightarrow [0,1]^n$ as
\begin{small}
\begin{equation}
\label{eq_transfer}
[f_p(\vq)]_i=\med{\{0,1,[g_p(\vq)]_i\}}
\end{equation}
\end{small}
in which
\begin{small}
$$[g_p(\vq)]_i = \frac{b_i-p+\sum_{j\in [n]}{T_{j,i}q_j}}{b_i-a_i}.$$
\end{small}
Notice that $\vq$ is an equilibrium at price $p$ if and only $f_p(\vq)=\vq$.
\end{definition}

Using Brouwer fixed point theorem, the existence of BNE is not surprising, even when influences are negative.
However, we will show in \appendixref{SEC:HARDNESS} that computing BNE
will be PPAD-hard
with negative influences.
We now define the pessimistic and optimistic
equilibria (similar to e.g. Van Zandt and Vives~\cite{Vives-VanZandt2007}) based on the transfer function.
\begin{definition}
  \label{def:poeq}
Let $f^{(1)}_p = f_p$, and $f_p^{(m)}(\vq)= f_p( f_p^{(m-1)}(\vq))$ for $m \geq 2$.
  When all influences are {\em non-negative}, we define
\begin{itemize}
\item{\bf Pessimistic equilibrium}: $\uq(p)=\lim_{m \rightarrow
    \infty}{f_p^{(m)}(\zero)}$;
\item{\bf Optimistic equilibrium}: $\oq(p)=\lim_{m \rightarrow
    \infty}{f_p^{(m)}(\one)}$.
\end{itemize}
\end{definition}

We remark that both limits exist by monotonicity of $f$ (see \factref{lemma_monotonicity} below), when all
influences are non-negative.
In addition,
 $\uq(p)$ and $\oq(p)$ are both equilibria themselves, because
 $f_p(\uq(p))=\uq(p)$ and $f_p(\oq(p))=\oq(p)$. We later show that
 $\uq(p)$ and $\oq(p)$ are the lower bound and upper bound for any
 equilibrium at price $p$ respectively.
Now we state some properties of equilibria, which we will use extensively later. Their proofs are in \appendixref{app:pre}.

For two vectors $\vv_1,\vv_2\in \mathbb{R}^n$,
 we write $\vv_1\geq \vv_2$ if $\forall i\in [n],\, [\vv_1]_i \geq
 [\vv_2]_i$ and we write $\vv_1 > \vv_2$ if $\vv_1 \geq \vv_2 \land
 \vv_1\neq \vv_2$.

\begin{fact}\label{lemma_monotonicity}
  When all influences are non-negative,
  given $p_1 \leq p_2, \vq^1 \leq \vq^2$, the transfer function
  satisfies $f_{p_2}(\vq^1) \leq f_{p_1}(\vq^1) \leq f_{p_1}(\vq^2)$.
\end{fact}

\begin{lemma}\label{lemma_eqproperty}
  When all influences are non-negative,
equilibria satisfy the following properties:
\begin{enumerate}[\ \ \ \rm a)]
  \item For any equilibrium $\vq$ at price $p$, we have $\uq(p) \leq \vq
    \leq \oq(p)$.
  \item Given price $p$, for any probability vector $\vq \leq \uq(p)$,
   we have $f_p^\iter(\zero)=\uq(p)=f_p^{(\infty)}(\vq)$.
  \vspace{1mm}
  \item Given price $p_1 \leq p_2$, we have $\uq(p_1) \geq \uq(p_2)$ and
    $\oq(p_1) \geq \oq(p_2)$.
  \vspace{1mm}
\item $\uq(p) = \lim_{\varepsilon \rightarrow 0+} \uq(p+\varepsilon)$ and
      $\oq(p) = \lim_{\varepsilon \rightarrow 0-} \uq(p+\varepsilon)$.
\end{enumerate}
\end{lemma}

In this paper, we consider the problem that whether we can exactly
calculate the pessimistic (resp. optimistic) equilibrium, and whether we
can maximize the revenue.  The latter is formally defined as follows:
\begin{definition}[Revenue maximization problem]
~\\
  Assume the value of agent $i$ is sampled from $U(a_i,b_i)$ and the
  influence matrix $\inft$ is given. The revenue maximization problem is to
  compute an optimal price with respect to the {\em pessimistic
    equilibrium} (resp. {\em optimistic equilibrium }):
\begin{small}
$$\argmax_{p >0} \sum_{i\in [n]} p\cdot [\uq(p)]_i
\mbox{\; (resp. }
\argmax_{p>0} \sum_{i\in [n]} p\cdot [\oq(p)]_i
\mbox{ )}.$$
\end{small}
\end{definition}

Notice that the optimal revenue with respect to the
pessimistic equilibrium is robust against equilibrium selection. By
\lemmaref{lemma_eqproperty}(a), no matter which equilibrium the
agents choose, this revenue is a minimal guarantee from the seller's
perspective. The revenue guarantees for pessimistic and optimistic equilibria is an important objective to study; see for instance the {\em price of anarchy} and the \emph{price of stability} in \cite{Nisan2007} for details.


\section{The Main Algorithm}
\label{SEC:BAYESIAN}

When all influences are non-negative, can we calculate $\uq(p)$ and $\oq(p)$ in polynomial
time? We answer this
question positively in this section by providing an efficient algorithm which
computes the optimal revenue as well as the $\uq(p)$ and $\oq(p)$ for
any price $p$.

\subsection{A counter example for iterative method}\label{app:method}
Before coming to our efficient algorithm, notice that it is possible to iteratively
apply the transfer function (\equationref{eq_transfer}) to reach the equilibria by definition.
However, this may require exponential number of steps to converge, as illustrated in the following example.
\begin{equation*}
\left\{
  \begin{split}
    &p=1 \\
    &[a_1,b_1]=[0,2], \ \ [a_i,b_i]=[0,1](2 \leq i \leq n) \\
    &\inft_{i,i+1}=0.5(1 \leq i \leq n-2), \ \ \inft_{n-1,n}=\inft_{n,n-1}=1, \ \ \textrm{other } \inft_{j,i}=0
  \end{split}
\right.
\end{equation*}
we can obtain that \[f_p^{(n-2)}(\zero)=(1/2,1/2^2,...,1/2^{n-2},0,0)\]
We proceed the calculation:
\begin{small}
\begin{equation*}
\left\{
  \begin{array}{lll}
    f_p^{(n-2+2k)}(\zero) &=(1/2,1/2^2,...,1/2^{n-2},k/2^{n-1},k/2^{n-1}), &\mathrm{if}\quad 0 \leq k \leq 2^{n-2} \\
    f_p^{(n-2+2k+1)}(\zero) &=
    (1/2,1/2^2,...,1/2^{n-2},(k+1)/2^{n-1},k/2^{n-1}), &\mathrm{if}\quad 0 \leq k < 2^{n-2} \\
    f_p^\iter(\zero) &= (1/2,1/2^2,...,1/2^{n-2},1,1)
  \end{array}
\right.
\end{equation*}
\end{small}

It can be seen from above that it takes $\Omega(2^n)$ number of steps before we reach the fixed point.

\subsection{Outline of our line sweep algorithm}
\label{SUBSEC:OUTLINE}



We start to introduce our algorithm with the easy case where valuations of agents are fixed.  Consider
the pessimistic decision vector $\ud(p)$ as a function of $p$.  By
monotonicity, there are at most $O(n)$ different such vectors when $p$
varies from $+\infty$ to $0$. In particular, at each price $p$, if we decrease $p$ gradually to some threshold value, one more agent would change his decision to buy the product.  Naturally, such kind of process can be casted in
the ``line sweep algorithm'' paradigm.


When the private valuations of the agents are sampled from uniform distributions, the
line sweep algorithm is much more complicated.  We now introduce the
algorithm to obtain the pessimistic equilibrium $\uq(p)$, while the
method to obtain $\oq(p)$ is similar.\footnote{We sweep the price from
  $+\infty$ to $0$ to compute the pessimistic equilibrium, but we need
  to sweep from $0$ to $+\infty$ for the optimistic one.}  The essence
of the line sweep algorithm is processing the events corresponding to
some structural changes. We define the possible structures of a
probability vector as follows.

\begin{definition}\label{def_structure}
  Given  $\vq \in [0,1]^n$, we define the structure
  function $S:[0,1]^n \rightarrow \{0,\star,1\}^n$ satisfying:
\begin{small}
\begin{equation}
[S(\vq)]_i=\left\{
\begin{array}{ll}
0, &q_i=0 \\
\star, &q_i \in (0,1) \\
1, &q_i = 1.
\end{array}
\right.
\end{equation}
\end{small}
\end{definition}

Our line sweep algorithm is based on the following fact: when $p$ is
sufficiently large, obviously $\uq(p)=\zero$; with the decreasing of
$p$, at some point $p=p_1$ the pessimistic equilibrium $\uq(p)$ becomes
non-zero, and there exists some {\em structural change} at this moment.
Due to the monotonicity of $\uq(p)$ in \lemmaref{lemma_eqproperty},
 such structural changes can happen at most $2n$ times.
(Each agent $i$ can contribute to at most two changes: $0\rightarrow
\star$ and $\star \rightarrow 1$.)
Therefore, there exist threshold prices $p_1 > p_2 > \cdots > p_m$ for
$m\leq 2n$ such that within two consecutive prices,
        the structure of the pessimistic equilibrium remains unchanged
        and $\uq(p)$ is a linear function of $p$.
This indicates that the total revenue,
i.e.,
 $p \cdot \sum_i{[\uq(p)]_i}$,
and its maximum value is easy to obtain.
If we can compute the threshold prices and the corresponding pessimistic
equilibrium $\uq(p)$ as a function of $p$, it will be straightforward to
determine
the optimal price $p$.  

There are several difficulties to address in this line sweep algorithm.
\begin{itemize}
\item First, degeneracies,
i.e., more than one structural changes in one event,
are intrinsic in our problem.
Unlike geometric problems where degeneracies can often be eliminated by
perturbations, the degeneracies in our problem are persistent
to small perturbations.
\item Second, to deal with degeneracies, we
need to identify the next structural change, which is related to the
eigenvector corresponding to the largest eigenvalue of a linear
operator. By a careful inspection, we avoid solving eigen systems so
that our algorithm can be implemented by pure algebraic computations.
\item Third, after the next change is identified, the usual method of
        pushing the sweeping line further does not work directly in our
        case. Instead, we recursively solve a subproblem and
        combine the solution
of the subproblem with the current one to a global solution. The
polynomial complexity of our algorithm is guaranteed by the monotonicity
of the structures.
\end{itemize}

We first design a line sweep algorithm for the problem with a diagonal
dominant condition, which will not contain degenerate cases, in
\sectionref{sec:restricted}. Then we describe techniques to deal with the
unrestricted case in \sectionref{sec:unrestricted}.

\subsection{Diagonal dominant case}
\label{sec:restricted}


\begin{definition}[Diagonal dominant condition]
~\\
Let $L_{i,j}=T_{j,i}/(b_i-a_i)$ and  $L_{i,i}=\inft_{i,i}=0$.
The matrix $I-L$ is {\em strictly diagonal dominant}, if
$\sum_j{L_{i,j}}=\sum_j{\inft_{j,i}/(b_i-a_i)}<1$.
\end{definition}

This condition has some natural interpretation on the buying behavior of
the agents.
It means that the decision of any agent cannot be solely determined by
the decisions of her friends.
In particular, the following two situations cannot occur {\em simultaneously}
         for any agent $i$ and price $p$:
        a) agent $i$ will not buy the product regardless of her
        own valuation when none of her friends
        bought the product($p \ge b_i$), and
        b) agent $i$ will always buy the product regardless
        of her own valuation when all her friends bought
        the product ($\sum_jT_{j,i} + a_i \ge p$).


In our line sweep algorithm, we maintain a {\em partition} $Z \cup W \cup
O=V=[n]$, and name $Z$ the {\em zero set}, $W$ the {\em working set} and
$O$ the {\em one set}. This corresponds to the structure $\vs \in \{0,\star,1\}^n$ as follows:
$$ s_i=0\,(\forall i \in Z), \ \ \ s_i=\star\,(\forall i \in W), \ \ \
s_i=1\,(\forall i \in O).$$ We use $\mathbf{x}_W$ or $[\mathbf{x}]_W$ to denote the restriction
of vector $\mathbf{x}$ on set $W$, and for simplicity we write
$\langle\mathbf{x}_Z,\mathbf{x}_W,\mathbf{x}_O\rangle=\mathbf{x}$.  Let
$L_{W\times W}$ be the projection of matrix $L$ to $W \times W$, and
$f|_W$ be the restriction of function $f$ on $W$.

We start from the price $p = +\infty$ where the structure of the
pessimistic equilibrium $\uq(p)$
is
$\vs^0=\zero$, i.e., $Z =
[n]$ and $W = O = \emptyset$.  The first event happens when $p$ drops to
$p_1=\max_
i
{b_i}$ and
$\uq(p)$ starts to
become non-zero.
%

Assume now we have reached threshold price $p_t$, the current
pessimistic equilibrium is $\vq^t=\uq(p_t)$, and the structure in interval $(p_t,p_{t-1})$ (or $(p_t,+\infty)$ if $t=1$) is
$\vs^{t-1}$.
We  define
\begin{small}
\begin{align*}
\label{eq:xy}
\vx&=\left(\frac{b_1-p_t}{b_1-a_1},\frac{b_2-p_t}{b_2-a_2},\ldots,
\frac{b_n-p_t}{b_n-a_n}\right)^T, \mbox{ and } 
\vy=\left(\frac{1}{b_1-a_1},\frac{1}{b_2-a_2},
\cdots,\frac{1}{b_n-a_n}\right)^T.
\end{align*}
\end{small}



To analyze the pessimistic equilibrium in the next price interval,
 for price $p=p_t-\varepsilon$ where $\varepsilon>0$,
 we write function $g_p(\cdot)$ (recall \equationref{eq_transfer}) as:
$$g_{p_t-\varepsilon}(\vq)=\vx+\varepsilon\vy+L\vq.$$

 For $p \in (p_t,p_{t-1})$, let  partition $Z \cup W \cup O = [n]$
        be consistent with the structure $\vs^{t-1}$.  According to
 \definitionref{def_structure} 
 and the right continuity $\vq^t = \lim_{p
   \rightarrow p_t+}{\uq(p)}$ (see \lemmaref{lemma_eqproperty}d), we have
\begin{equation}
\begin{array}{ll}
\forall i \in Z, &[g_{p_t}(\vq^t)]_i=[\vx+L\vq^t]_i \leq 0 \\
\forall i \in W, &[g_{p_t}(\vq^t)]_i=[\vx+L\vq^t]_i \in (0,1] \\
\forall i \in O, &[g_{p_t}(\vq^t)]_i=[\vx+L\vq^t]_i \geq 1
\end{array}
\end{equation}
\paragraph{Step 1:}
{\it For any $i \in Z$, if $[\vx+L\vq^t]_i=0$, move $i$ from zero set
  $Z$ to working set $W$;
      for any $i \in W$, if $[\vx+L\vq^t]_i=1$, move $i$ from working
      set $W$ to one set $O$.}

Notice that the structural changes we apply in Step 1 are exactly the
changes defining the threshold price $p_t$.
We will see in a moment that after the process in Step 1, the new
partition will be the next structure $\vs^t$ for $p \in
(p_{t+1},p_t)$. In other words, there is no more structural change at
price $p_t$.


In the next two steps, we calculate the next threshold price $p_{t+1}$.
For notation simplicity, we assume $Z,W$ and $O$ remain unchanged in these two steps.
When $p$ decreases by $\varepsilon$, we
show that the probability vector of agents in $W$,
$[\uq(p)]_W$, increases linearly with respect to $\varepsilon$. (See
$\vr_W(\varepsilon)$ below.)
However, this linearity holds until we reach some point,
 where the next structural change takes place.


\paragraph{Step 2:}
{\it Define the vector $\vr(\varepsilon) \in \mathbb{R}^n$, and let:}
\begin{small}
\begin{equation}
\begin{array}{rl}\label{eq_step2_1}
\vr_W(\varepsilon) &= \varepsilon(I-L_{W \times W})^{-1} \vy_W + \vq_W^t \\
&=\varepsilon (I-L_{W \times W})^{-1} \vy_W + [\vx+L\vq^t]_W\\
\vr_Z(\varepsilon) &= \vx_Z + \varepsilon \vy_Z + L_{Z \times
  W}\vr_W(\varepsilon) + L_{Z \times O} \one_O \\
&= \varepsilon(\vy_Z + L_{Z \times W}(I-L_{W \times W})^{-1}\vy_W) + [\vx + L\vq^t]_Z \\
\vr_O(\varepsilon) &= \vx_O + \varepsilon \vy_O + L_{O \times
  W}\vr_W(\varepsilon) + L_{O \times O} \one_O \\
&= \varepsilon(\vy_O + L_{O \times W}(I-L_{W \times W})^{-1}\vy_W) + [\vx + L\vq^t]_O
\end{array}
\end{equation}
\end{small}

Clearly, $\vr(\varepsilon)$ is linear to
$\varepsilon$ and we write $\vr(\varepsilon)=\varepsilon \vl + (\vx + L \vq^t)$
where $\vl =\langle \ell_1,\ell_2,\ldots, \ell_n \rangle \in \mathbb{R}^n$ is the linear coefficient derived from
        \equationref{eq_step2_1}.
When $I-L$ is strictly diagonal dominant, the largest eigenvalue of $L_{W\times W}$ is smaller than $1$.
Using this property one can verify (see \lemmaref{lemma_step23}) that
$\vl$ is strictly positive.

{\it
\paragraph{Step 3:}
\begin{small}
\begin{equation}\label{eq_step3}
\varepsilon_{min}=\min{\left\{ \min_{i \in Z}{\left\{
        \frac{0-[\vx+L\vq^t]_i}{\ell_i}\right\}}, \min_{i\in W}\left\{
      \frac{1-[\vx+L\vq^t]_i}{\ell_i} \right\} \right\}}
\end{equation}
\end{small}
}

Using the positiveness of vector $\vl$ one can verify that $\varepsilon_{min}>0$ (see \lemmaref{lemma_step23}).
We show that the next threshold price $p_{t+1} = p_t
-\varepsilon_{min}$ by the following lemma. The proof is in the \appendixref{app:bayesian}.



\begin{lemma}\label{lemma_alg1}
  $\forall 0 < \varepsilon \leq \varepsilon_{min}$, $\uq(p_t-\varepsilon) =
  \langle \zero_Z, \vr_W(\varepsilon), \one_O \rangle$.
\end{lemma}

We remark here that the above lemma has confirmed that our structural adjustments in Step 1 are correct and complete.  Now we let
$p_{t+1}=p_t-\varepsilon_{min}, \vq^{t+1}=\langle \zero_Z,
\vr_W(\varepsilon_{min}), \one_O \rangle$.  The next structural change will
take place at $p=p_{t+1}$.  This is because according to the definition
of $\varepsilon_{min}$ (\equationref{eq_step3}), there must be some
\begin{align*}
i\in W \wedge [\vx+\varepsilon_{min}\vy+L\vq^{t+1}]_i=1,
\text{ \,or \, }
i\in Z\wedge [\vx+\varepsilon_{min}\vy+L \vq^{t+1}]_i=0.
\end{align*}
One can see that in the next iteration, this $i$ will move to one set
$O$ or working set $W$ accordingly.
Therefore, we can iteratively execute the above three steps by
sweeping the price further down. For completeness, we attach the
pseudocode in \algorithmref{alg_alg1} in \appendixref{app:bayesian}.

The return value of our constrained line sweep method is a function
$\uq$ which gives the pessimistic equilibrium for any price $p\in
\mathbb{R}$,
 and $\uq(p)$ is a piecewise linear function of $p$ with no more than
 $2n+1$ pieces.
%
%
%
All three steps in our algorithm can be done in polynomial time. Since
there are only $O(n)$ threshold prices, we have the following result.

\begin{theorem}
  When the matrix $I-L$ is {\em strictly diagonal dominant},
  \algorithmref{alg_alg1} calculates the pessimistic equilibrium $\uq(p)$  (resp. $\oq(p)$)
  for any given price $p$ in polynomial time, together with the
  optimal revenue.
\end{theorem}





\subsection{General case}
\label{sec:unrestricted}

After relaxing the diagonal dominance condition, the algorithm
becomes more complicated. This can be seen from this simple scenario.
There are $2$ agents, with $[a_1,b_1]=[a_2,b_2]=[0,1]$, and
$\inft_{1,2}=\inft_{2,1}=2$.  One can verify that $\uq(p)=(0,0)^T$ when $p \geq
1$; $\uq(p)=(1,1)^T$ when $p<1$.

In this example, there is an {\em equilibrium jump} at price $p=1$,
i.e., $\uq(1)\neq \lim_{p \rightarrow 1-}{\uq(p)}$.
\algorithmref{alg_alg1} essentially
requires
that both the left and the
right continuity of $\uq(p)$. However, only
the right continuity is unconditional by
\lemmaref{lemma_eqproperty}d.  More importantly,  degeneracies may
occur: the new structure
$\vs^t$ when $p=p_t$ cannot be determined all in once in Step 1.
When $p$ goes from $p_t+\varepsilon$ to $p_t-\varepsilon$, there might take
place even two-stage jumps: some index $i$ might leave  $Z$ for $O$,
without being in the intermediate state.

Let $\rho(L)$ be the largest norm of the eigenvalues in matrix $L$.
The ultimate reason for
such degeneracies, is $\rho(L_{W\times W})\geq 1$ and $ (I-L_{W\times
  W})^{-1} \neq \lim_{m \rightarrow \infty}(I+L_{W\times
  W}+\cdots+L_{W\times W}^{m-1})
$. 
%
 We will prove shortly in such cases,
 those structural changes in Step 1 are {\em incomplete}, that is, as $p$
 sweeps across $p_t$, at least one more structural change will take
 place. We derive a method to identify one {\em pivot}, i.e. an additional
         structural
 change, in polynomial time. Afterwards, we recursively solve a
 subproblem with set $O$ taken out, and combine the solution from the subproblem with the current one.
 The follow lemma shows that whether $\rho(L) < 1$ can be determined
        efficiently.



\begin{lemma}
\label{lem:detect_eigen}
  Given non-negative matrix $M$, if $I-M$ is reversible and $(I-M)^{-1}$
  is also non-negative, then $\rho(M)<1$; on the contrary, if $I-M$ is
  degenerate or if $(I-M)^{-1}$ contains negative entries, $\rho(M)\geq
  1$.
\end{lemma}


\subsubsection{Finding the pivot.}
When $\rho(L_{W\times W})<1$ for the new working set $W$, one can find
the next threshold price $p_{t+1}$ following Step $2$ and $3$ in the previous
subsection. Now, we deal with the case that $\rho(L_{W\times
  W})\ge 1$ by showing that there must exists some
additional agent $i \in W$ such that $[\uq(p)]_i = 1$ for any $p$
smaller than the current price. We call such agent a {\em pivot}.


Since $\rho(L_{W\times W})\geq 1$, we can always find a
non-empty set $W_1\subset W$ and $W_2 = W_1 \cup \{w\} \subset W$,
satisfying $\rho(L_{W_1\times W_1})<1$ but $\rho(L_{W_2\times W_2})\geq
1$.  The pair $(W_1,W_2)$
can be found by ordering the elements
in $W$ and add them to $W_1$ one by one.
We now show that there is a pivot in $W_2$.

  As $L_{W_2\times W_2}$ is a non-negative matrix, based on
  \lemmaref{lemma_matrix} there exists a non-zero eigenvector $\vu_{W_2} \geq
  \zero_{W_2}$ such that $L_{W_2\times W_2}\vu_{W_2} = \lambda\vu_{W_2}$ and
  $\lambda = \rho(L_{W_2\times W_2})\geq 1$.  $\vu_{W_2}$ can be extended to $[n]$ by
  defining $\vu_{[n]\setminus W_2}=\zero_{[n]\setminus W_2}$.
  Let
  \begin{equation}\label{eq_pick}
    k=\mathop{\rm arg\,min}_{k\in W_2, u_k\neq0}{\frac{1-q_k^t}{u_k}}
    =\mathop{\rm arg\,min}_{k\in [n], u_k\neq0}{\frac{1-q_k^t}{u_k}}
  \end{equation}

 Now we
prove that $k$ is a pivot. Intuitively, if we slightly increase the
 probability vector $\vq^t_{W_2}$ by $\delta\vu_{W_2}$, where $\delta$ is a
 small constant, by performing the transfer function only on agents in $W$ $m$ times,
  their probability will increase by $\delta (1+\lambda+..+\lambda^m) \vu_{W_2}$, while
 $\lambda \ge 1$. Therefore, after performing the transfer function
 sufficiently many times, agent $k\in W_2$'s probability will hit $1$
 first.


%
%
%

\begin{lemma}
\label{lem:pivot}
$\forall W_2\subset W$ s.t. $\rho(L_{W_2 \times W_2})\geq 1$, we have
$\forall \varepsilon > 0$,
 $[\uq(p_t-\varepsilon)]_k=1$.
\end{lemma}

We remark that if we can exactly estimate the eigenvector (which may be
irrational), then the above lemma has already determined that the $k$ defined in \equationref{eq_pick} is a pivot.
To avoid the eigenvalue computation, we find a quasi-eigenvector $\vu$ in
the following manner.
\begin{equation}\label{eq_newu}
\vu = \left\{\begin{split}
&\vu_{W_1} = (I-L_{W_1\times W_1})^{-1}L_{W_1 \times \{w\}}; \\
&u_{w}=1; \\
&\vu_{Z\cup O \cup W \setminus W_2} = \zero_{Z \cup O \cup W \setminus W_2}.
\end{split}\right.
\end{equation}
The meaning of the above vector is as follows.  If we raise agent
$w$'s probability by $\delta$, those probabilities of agents in $W_1$
increase proportionally to $L_{W_1 \times \{w\}}\delta$.  Assuming that we
ignore the probability changes outside $W_2$ (which will even increase
the probabilities in $W_2$), the probability of
agents in $W_1$ will eventually converge to
$(I+L_{W_1\times W_1}+L_{W_1\times W_1}^2+...)L_{W_1 \times
   \{w\}}\delta = (I-L_{W_1\times W_1})^{-1}L_{W_1 \times
   \{w\}}\delta.$

We will see that the real probability vector increases at least ``as
much as if we increase in the direction of $\vu$''.
In other words, we pick a pivot in the same way as
\equationref{eq_pick}. The following is the critical lemma to support our
result.


\begin{lemma}\label{lemma_alg2_new_1}
Given the definition of $\vu$ in \equationref{eq_newu} and
$k$ using
\equationref{eq_pick}, we have $\forall \varepsilon > 0,
[\uq(p_t-\varepsilon)]_k=1.$
\end{lemma}

\subsubsection{Recursion on the subproblem.}
Let $W'=W\setminus\{k\}$, $O'=O\cup\{k\}$, and we consider a subproblem
with $n'=n-|O'|<n$ agents, where $k$ is the pivot identified in the previous
section.  This subproblem is a projection of the original one, assuming
that the agents in $O'$ always tend to buy the product.
\begin{equation}\label{eq_alg2_1}
\textstyle
\forall i\in Z\cup W', \quad [a_i',b_i']=[ a_i+\sum_{j\in
    O'}{T_{j,i}}, b_i+\sum_{j \in O'}{T_{j,i}} ].
\end{equation}
By recursively solving this new instance, we can solve the pessimistic
equilibrium of the subproblem for any given price $p$.  This recursive
procedure will eventually terminate because every invocation reduces the
number of agents by  at least $1$.  The following lemma tells us that for
any $p<p_t$, the pessimistic equilibrium of the original problem and the
subproblem are one-to-one.  

\begin{lemma}\label{lemma_alg2_2}
Let $\uq'(p)$ be the pessimistic equilibrium function in the
subproblem. We have:
$$\forall p<p_t, \uq(p)=\langle\uq'(p),\one_{O'}\rangle.$$
\end{lemma}

At this moment we have solved the pessimistic equilibrium $\uq(p)$ for
$p<p_t$, and thus solved the original problem.  We summarize our
unrestricted line sweep method in \algorithmref{alg_alg2} in  \appendixref{app:bayesian} for completeness.
Again $\uq(p)$ is a piecewise linear function of $p$ with no more than
$2n+1$ pieces.  

\begin{theorem}
  For matrix $T$ satisfying $T_{i,i}=0$ and $T_{i,j}\geq 0$, in polynomial
  time \algorithmref{alg_alg2} is able to calculate the pessimistic
  equilibrium $\uq(p)$ (resp. $\oq(p)$) at any price $p$, along with the
  optimal $p$ that ensures the maximal revenue under the pessimistic
  equilibrium (resp. the optimistic equilibrium).
\end{theorem}
\section{Extensions}\label{SEC:OTHER}
We discuss some possible extensions of our model in this section with
both positive and negative influences.
When the influence values can be negative, it is actually PPAD-hard to compute an {\em approximate}
equilibrium. We define a probability vector $\vq$ to
be an $\varepsilon$-approximate
equilibrium for price $p$ if:
$$q_i \in (q_i'-\varepsilon, q_i'+\varepsilon),$$
where $q_i'=\med\{0,1,
\frac{b_i-p+\sum_{j\in[n]}T_{j,i}q_j}{b_i-a_i}\}$. We have the following theorem, whose proof is deferred to \appendixref{SEC:HARDNESS}.
\begin{theorem}
\label{thm:hard}
It is PPAD-hard to compute an $n^{-c}$-approximate equilibrium of our
pricing system for any $c>1$ when influences can be negative.
\end{theorem}

In discriminative pricing setting, we study the revenue maximization
problem in two natural models.
We assume the agents are partitioned into $k$ groups. The seller can
offer different prices to different groups. The first model we
consider is the fixed partition model, i.e.,  the partition is predefined. In the
second model, we allow the seller to partition the agents into $k$
groups and offer prices to the groups respectively. We have the following two theorems, whose proofs are deferred to \appendixref{SEC:MULTIPRICE}.

\begin{theorem}
\label{thm:fptas}
There is an FPTAS for the discriminative pricing problem in the fixed
partition case with constant $k$.
\end{theorem}

\begin{theorem}
\label{thm:prehard}
It is NP-hard to compute the optimal pessimistic discriminative
pricing equilibrium in the choosing partition case.
\end{theorem}

\clearpage

\appendix
\begin{center}
{\huge Appendix}
\end{center}

\section{Missing Proofs in \sectionref{SEC:PRE}}\label{app:pre}

Before proving \lemmaref{lem:eq_eq}, let us recall the {\em Bayesian
  Nash Equilibrium} (BNE) from the economics literature (see e.g.
Chapter 8 of~\cite{Mas-Colell1995}). Formally, in a Bayesian game, each
agent has a private type
$v_i \in \Omega_i$, where $\Omega_i$ is the set of all possible types
of agent $i$. Let $\mathcal{S}_i$ be the action space for agent
$i$.
Slightly abusing the notation, we define the (pure)
strategy of agent $i$ as a function $d_i:\, \Omega_i \rightarrow
\mathcal{S}_i$.  The utility of agent
$i$ when the type configuration $\vv$ is known
is $\mathcal{U}_i( \langle d_i(v_i),\vd_{-i}(\vv_{-i}) \rangle,v_i)$, where
$\vd_{-i}(\vv_{-i})$ is the joint actions of all agents other than
$i$. Define the expected utility of agent $i$ as:
$$\tilde{\mathcal{U}_i}(d_1(\cdot),\ldots, d_n(\cdot)) :=
\mathbb{E}_{\vv \sim \Omega_1\times \cdots \times \Omega_n}
[\mathcal{U}_i( \langle d_i(v_i),\vd_{-i}(\vv_{-i}) \rangle, v_i)],
$$
where the expectation is taking over all type configurations of the
agents.


\begin{definition}[Bayesian Nash Equilibrium ({\bf BNE})]
\label{def:bne}
A profile of strategies $\{d_1(\cdot),\ldots,d_n(\cdot)\}$ is a (pure)
Bayesian Nash Equilibrium, if and only if, for all $i$, all $v_i
\in \Omega_i$ and any other strategy $d_i'(\cdot)$ of agent $i$, such
that,
\begin{small}
$$\tilde{\mathcal{U}_i}(d_1(\cdot),\ldots,d_i(\cdot),\ldots, d_n(\cdot))
\geq \tilde{\mathcal{U}_i}(d_1(\cdot),\ldots,d_i'(\cdot),\ldots, d_n(\cdot))
$$
\end{small}
\end{definition}

In our setting, $\Omega_i$ is the set of private values of agent $i$
and $d_i(\cdot)$ maps a particular value $v_i$ to $\{0,1\}$.
The utility of agent $i$ is given in \equationref{eqn:utility}.
Notice that mixed strategies are almost irrelevant here, because while fixing other agent's private valuations, agent $i$'s strategy is a simple choice between to buy or not to buy. Unless the utility function $u_i(S,p)=0$, there is always a unique better choice for her.

For ease of
presentation, we redefine the equilibrium based on the buying
probability of the agents and show that they are equivalent.

\medskip
{\bf \noindent \lemmaref{lem:eq_eq} (restated).}\quad
{\sl
  Given equilibrium $\vq$ (recall \definitionref{def_probeq}), the strategy profile such agent $i$ ``buys
  the product if and only if her internal valuation $v_i\geq p-\sum_{j\neq
    i}{T_{j,i}q_j}$'' is a Bayesian Nash equilibrium; on the contrary, if
  a strategy profile is a Bayesian Nash equilibrium, then the
  probability that agent $i$ buys the product satisfies \equationref{eq_probeq0}.
}
\begin{proof}
  Let strategy profile $\vd(\cdot)=(d_1(\cdot),d_2(\cdot),...,d_n(\cdot))$ be
  a Bayesian Nash equilibrium, and
  $q_i=\prob_{v_i}[d_i(v_i)=1]$ be the probability that agent $i$ buys
  the product under this profile.  In our setting, the utility of agent $i$ is defined by
  \equationref{eqn:utility}.
  Now we calculate the expected utility of agent $i$:
\begin{equation}\label{eq_bne_utility}
\begin{split}
  \tilde{u}_i(d_i(\cdot),\vd_{-i}(\cdot))
  &= \mathbb{E}_{v_i}[d_i(v_i) \cdot
  (v_i-p+\mathbb{E}_{\vv_{-i}}[\sum_{j\neq i}{\inft_{j,i}d_j(v_j)}])] \\
  &= \mathbb{E}_{v_i}[d_i(v_i) \cdot (v_i-p+\sum_{j \neq i}{\inft_{j,i}q_j})]
\end{split}
\end{equation}
To satisfy the condition of Bayesian Nash equilibrium, we must have that
$\forall d_i'(\cdot)$, $\tilde{u_i}(d_i(\cdot),\vd_{-i}(\cdot)) \geq
\tilde{u_i}(d_i'(\cdot),\vd_{-i}(\cdot))$.  This means, $d_i(v_i)$ must
be $1$ whenever $v_i-p+\sum_{j\neq i}{\inft_{j,i}q_j}$ is positive, and $0$
whenever it is negative \footnote{Strictly speaking, we should say
  ``almost everywhere'' but this does not affect our analysis.}.
Therefore, $q_i=\prob[d_i(v_i)=1]=\prob[v_i-p+\sum_{j\neq
  i}{\inft_{j,i}q_j}>0]$, satisfying \definitionref{def_probeq}.

On the contrary, the strategy that agent $i$ ``buys whenever
$v_i\geq p-\sum_{j\neq i}{\inft_{j,i}q_j}$'' can be denoted as
$d_i(v_i)=\mathbb{I}[v_i-p+\sum_{j \neq i}{\inft_{j,i}q_j}>0]$ where
$\mathbb{I}$ is the indicator function.  This obviously maximizes
\equationref{eq_bne_utility}, and is a Bayesian Nash equilibrium.
\end{proof}


\medskip
{\bf \noindent \lemmaref{lemma_eqproperty} (restated).}\quad
{\sl
Equilibria satisfy the following properties:
\begin{enumerate}[\ \ \ \rm a)]
  \item For any equilibrium $\vq$ at price $p$, we have $\uq(p) \leq \vq
    \leq \oq(p)$.
  \item Given price $p$, for any probability vector $\vq \leq \uq(p)$,
   we have $f_p^\iter(\zero)=\uq(p)=f_p^{(\infty)}(\vq)$.
  \item Given price $p_1 \leq p_2$, we have $\uq(p_1) \geq \uq(p_2)$ and
    $\oq(p_1) \geq \oq(p_2)$.
\item $\uq(p) = \lim_{\varepsilon \rightarrow 0+} \uq(p+\varepsilon)$ and
      $\oq(p) = \lim_{\varepsilon \rightarrow 0-} \uq(p+\varepsilon)$.
\end{enumerate}
}
\begin{proof} $ $
\begin{enumerate}[\ \ \ \rm a)]
\item By the definition of equilibrium, $\vq =
  f_p(\vq)=f_p^{(\infty)}(\vq)$.  Next according to
  $\zero\leq\vq\leq\one$ and the monotonicity of $f_p$, we derive that:
     $$ f_p(\zero) \leq f_p(\vq) \leq f_p(\one) \Rightarrow ... \Rightarrow f_p^{(\infty)}(\zero) \leq f_p^{(\infty)}(\vq) \leq f_p^{(\infty)}(\one).$$
   \item By symmetry we only need to prove the first half.  We already
     know that $f_p(\uq(p))=\uq(p)$, then recall the monotonicity of
     $f_p$
     \begin{equation*}
     \begin{split}
       \zero \leq \vq \leq \uq(p) &\Rightarrow f_p(\zero) \leq f_p(\vq) \leq f_p(\vq(p)) \Rightarrow ... \\
                                  &\Rightarrow f_p^\iter(\zero) \leq f_p^\iter(\vq) \leq f_p^\iter(\uq(p)) \\
                                  &\Rightarrow \uq(p) \leq f_p^\iter(\vq) \leq \uq(p).
     \end{split}
     \end{equation*}
     Notice that the last ``$\Rightarrow$'' is due to $f_p^\iter(\zero)=\uq(p)=f_p(\uq(p))=...=f_p^\iter(\uq(p))$,
      while the convergence of the limit $f_p^\iter(\vq)=\lim_{m \rightarrow \infty}{f_p^\iter(\vq)}$
      is ensured by the sandwich theorem.
    \item This time we use the combined monotonicity of the function $f$
      (\factref{lemma_monotonicity})
     \begin{equation*}
     \begin{split}
      p_1 \leq p_2 \wedge \zero \geq \zero &\Rightarrow f_{p_1}(\zero) \geq f_{p_2}(\zero) \\
      p_1 \leq p_2 \wedge f_{p_1}(\zero) \geq f_{p_2}(\zero) &\Rightarrow f_{p_1}^{(2)} \geq f_{p_2}^{(2)}(\zero) \\
      &... \\
      &\Rightarrow f_{p_1}^\iter(\zero) \geq f_{p_2}^\iter(\zero) \\
      &\Rightarrow \uq(p_1) \geq \uq(p_2)
     \end{split}
     \end{equation*}
     For similar reason we also have $\oq(p_1) \geq \oq(p_2)$.
\item We only prove the first half while the property of $\oq(p)$ can be obtained in similar way.
We first claim that for any fixed $m$, $f_p^{(m)}(\zero) =
  \lim_{\varepsilon\rightarrow 0+}f^{(m)}_{p+\varepsilon}(\zero)$.
  Since $f_p(\vq)$ is a continuous multi-variable function with respect to $(p,\vq)$, the composition $f_p^{(m)}(\vq)$ is also continuous.
  This directly implies our claim.
%
%
%

Now assume Property (d) is not true: there exists $\delta>0$ and $\varepsilon_0$ such that $\forall 0<\varepsilon<\varepsilon_0$,
 $[\uq(p) - \uq(p+\varepsilon)]_i > \delta$ for some $i$.
 By definition of the pessimistic equilibrium, there exists $m_0$ such that $[\uq(p) - f_p^{(m_0)}(\zero)]_i <\delta/2$.
 On the other hand by our claim just proved, we can choose $\varepsilon$ small enough such that
 $[f^{(m_0)}_p(\zero) - f_{p+\varepsilon}^{(m_0)}(\zero)]_i < \delta/2$.
 Combining the two we have $\delta>[\uq(p) - f_{p+\varepsilon}^{(m_0)}(\zero)]_i \geq [\uq(p) - \uq(p+\varepsilon)]_i$,
 where the second inequality is due to non-decreasing sequence
  $\{f_{p+\varepsilon}^{(m)}(\zero)\}_{m\geq1}$ that converge to $\uq(p+\varepsilon)$.
 This contradiction completes the proof.
 We remark here that the left continuity does not hold, see the beginning of \sectionref{sec:unrestricted}.
%
\end{enumerate}
\end{proof}

\section{Missing Proofs in \sectionref{SEC:BAYESIAN}}\label{app:bayesian}

Before proving \lemmaref{lemma_alg1}, we first show \equationref{eq_step3} is
well defined.

\begin{lemma}\label{lemma_step23}
$\vl \in \mathbb{R}_{+}^n$ and $\varepsilon_{min}>0$.
\end{lemma}
\begin{proof}
When $I-L$ is strictly diagonal dominant, the largest eigenvalue of $L_{W\times W}$ is smaller than $1$.
By the knowledge from complex analysis, the following limit exists
$$(I-L_{W\times W})^{-1} = I +L_{W\times W}+L_{W\times W}^2 +\cdots$$
and it is a non-negative matrix since $L$ is non-negative.

Now, $\vy$ is strictly positive and therefore $
  \vl_W = (I-L_{W \times W})^{-1}\vy_W \in \mathbb{R}_+^{|W|}$ is also positive.
Besides, recall the definition in \equationref{eq_step2_1} we have
$\vl_Z=\vy_Z + L_{Z\times W}\vl_W \in \mathbb{R}_+^{|Z|}$, $\vl_O =
\vy_O + L_{O\times W}\vl_W \in \mathbb{R}_+^{|W|}$, and therefore $\vl
\in \mathbb{R}_+^n$. Finally, by our Step 1, we have
$[\vx+L\vq^t]_i <0$ for $i\in Z$, and $[\vx+L\vq^t]_j <1$ for $j \in
W$. Therefore, $\varepsilon_{min}>0$ is properly defined.
\end{proof}

\medskip
{\bf \noindent \lemmaref{lemma_alg1} (restated).}\quad
{\sl
  $\forall 0 < \varepsilon \leq \varepsilon_{min}$, $\uq(p_t-\varepsilon) =
  \langle \zero_Z, \vr_W(\varepsilon), \one_O \rangle$.
}
\begin{proof}
We first show that $\vq = \langle \zero_Z, \vr_W(\varepsilon), \one_O \rangle$
is an equilibrium for $\varepsilon \in (0, \varepsilon_{min}]$.
By our definition of
$\varepsilon_{min}$, when $0 < \varepsilon \leq \varepsilon_{min}$ we must
have
\begin{small}
\begin{equation*}
\left\{
\begin{split}
[g_{p_t-\varepsilon}(\vq)]_W &= \vr_W(\varepsilon) \in [0,1]^{|W|} \\
[g_{p_t-\varepsilon}(\vq)]_Z &= \vr_Z(\varepsilon) \leq \zero_Z \\
[g_{p_t-\varepsilon}(\vq)]_W &= \vr_O(\varepsilon) \geq \vr_O(0) \geq \one_O
\end{split}
\right.
\end{equation*}
\end{small}
Since $f_{p_t-\varepsilon} = \med\{0,1,g_{p_t-\varepsilon}\}$
(\definitionref{def:transfer}), it must be the case that
$f_{p_t-\varepsilon}(\vq)=\vq$, i.e., $\vq$ is an equilibrium.
Next we lower bound the pessimistic equilibrium by $\uq(p_t-\varepsilon) \geq \vq$.
This will be sufficient to complete the proof following from \lemmaref{lemma_eqproperty}a.

Denote $p=p_t-\varepsilon$, notice that $\uq(p) = f_p^\iter(\zero)=f_p^\iter(\vq^t)$, where the second equality is because:
\begin{eqnarray*}
\textrm{\lemmaref{lemma_eqproperty}c} \Rightarrow \vq^t=\uq(p_t)\leq\uq(p) \\
\mathop{\Longrightarrow}^{\textrm{\lemmaref{lemma_eqproperty}b}} \uq(p)=f_p^\iter(\zero)=f_p^\iter(\vq^t)
\end{eqnarray*}
For the simplicity of notation, we define $\vx_W' := \vx_W + L_{W
\times O} \one_O$ as a constant vector, and according to the
definition of an equilibrium:
$$ \vq_W^t = \vx_W'+ L_{W\times W} \vq_W^t.$$
After repeated use of the monotonicity of transfer function $f$, we make
the following analysis \footnote{within which we implicitly adopted the
  following property:
\begin{eqnarray*} \varepsilon(I+L_{W \times W}+...+L_{W\times W}^{m-1})\vy_W + \vq_W^t
\leq \varepsilon_{min}(I-L_{W\times W})^{-1}\vy_W + \vq_W^t \leq
\one_W \end{eqnarray*}}
:
\begin{small}
\begin{equation}\label{eq_alg1}
\left\{
\begin{split}
  f_{p_t-\varepsilon}(\vq^t) &\geq \langle \zero_Z, \varepsilon \vy_W +
  \vq_W^t,
  \one_O \rangle \\
  f_{p_t-\varepsilon}^{(2)}(\vq^t) &\geq f_{p_t-\varepsilon}(\langle \zero_Z, \varepsilon \vy_W + \vq_W^t, \one_O \rangle) \\
  &\geq \langle \zero_Z, \vx_W'+\varepsilon \vy_W + L_{W\times
    W}(\varepsilon \vy_W + \vq_W^t),
  \one_O \rangle \\
  &
   =\langle \zero_Z, \varepsilon(I+L_{W \times W}) \vy_W + \vq_W^t, \one_O \rangle \\
  &\cdots \\
  f_{p_t-\varepsilon}^\iter(\vq^t) &\geq \langle \zero_Z,
  \varepsilon(\sum_{i=0}^{\infty}\left(L_{W
    \times W})^i\right)\vy_W + \vq_W^t, \one_O \rangle
\end{split}
\right.
\end{equation}
\end{small}
The last inequality in
\equationref{eq_alg1} implies that
\begin{eqnarray*}
\uq(p_t-\varepsilon) &= &f_{p_t-\varepsilon}^\iter(\vq^t) \geq \langle
\zero_Z, \varepsilon(I-L_{W \times W})^{-1}\vy_W+\vq_W^t, \one_O\rangle\\
&=&\langle \zero_Z, \vr_W(\varepsilon), \one_O \rangle = \vq
\end{eqnarray*}
\end{proof}

The following lemma in matrix analysis is important for our analysis.
\begin{lemma}[\cite{Horn1990}]\label{lemma_matrix}
  Given a non-negative matrix $M$ (i.e. $\forall i,j, M_{ij} \geq 0$),
  there exists a non-negative (and non-zero) eigenvector $\vx \geq\zero$
  satisfying $M\vx=\lambda\vx$, in which $\lambda = \rho(M)$ is a real
  number.
\end{lemma}

\medskip
{\bf \noindent \lemmaref{lem:detect_eigen} (restated).}\quad
{\sl
  Given non-negative matrix $M$, if $I-M$ is reversible and $(I-M)^{-1}$
  is also non-negative, then $\rho(M)<1$; on the contrary, if $I-M$ is
  degenerate or if $(I-M)^{-1}$ contains negative entry, $\rho(M)\geq
  1$.
}
\begin{proof}
  For the first half, assume the contrary that $\rho(M)\geq 1$.
  According to \lemmaref{lemma_matrix}, there exists a non-negative
  $\vx$ s.t. $(I-M)\vx=(1-\rho(M))\vx \leq \zero$.  As $(I-M)^{-1}$ is
  non-negative, multiply a non-positive vector $(I-M)\vx$ to its right is
  also non-positive: $(I-M)^{-1}(I-M)\vx \leq \zero$.  But the last
  inequality means that $x \leq \zero$ which contradicts the result in
  \lemmaref{lemma_matrix} saying $x$ is a non-negative and non-zero
  eigenvector.

  For the second half, if $I-M$ is degenerate then $(I-M)\vx=0$ has a
  non-zero solution, which implies $M\vx=\vx$ and $\rho(M)\geq 1$.
  Assume the contrary that $\rho(M)<1$, then $(I-M)^{-1}=I+M+M^2+...$ is
  non-negative as $M$ is non-negative, resulting in a contradiction.
\end{proof}

\begin{algorithm}
\caption{\textsf{ConstrainedLineSweepMethod}($n,T,\mathbf{a},\mathbf{b}$)}
\label{alg_alg1}
\begin{algorithmic}[1]
\REQUIRE $n,T,\mathbf{a},\mathbf{b}$.
\ENSURE The pessimistic equilibrium function $\uq : p \mapsto \uq(p)$.
\STATE $L_{i,j} \leftarrow \inft_{j,i}/(b_i-a_i)$;
\STATE $p_1 \leftarrow \max_{1 \leq i \leq n}{b_i}$;
\STATE $\uq(p)|_{[p_1,\infty)} \leftarrow \zero$;
\STATE $Z \leftarrow [n]$; $W \leftarrow \emptyset$; $O \leftarrow
\emptyset$; $t \leftarrow 1$;
\WHILE{$\uq(p_t) \neq \one$}

\STATE $\vq^t \leftarrow \uq(p_t)$;

\STATE $\vx \leftarrow ((b_1-p_t)/(b_1-a_1), (b_2-p_t)/(b_2-a_2), ...,
(b_n-p_t)/(b_n-a_n))$;
\STATE $\vy \leftarrow (1/(b_1-a_1), 1/(b_2-a_2), ..., 1/(b_n-a_n))^T$;

\FORALL{$i \in Z$ s.t. $\vx_i+\sum_j{L_{i,j}q_j^t}=0$} \label{code1:step1start}
\STATE $Z \leftarrow Z\setminus\{i\}$; $W \leftarrow W\cup\{i\}$;
\ENDFOR

\FORALL{$i \in W$ s.t. $\vx_i+\sum_j{L_{i,j}q_j^t}=1$}
\STATE $W \leftarrow W\setminus\{i\}$; $O \leftarrow O\cup\{i\}$;
\ENDFOR \label{code1:step1end}

\STATE $\vl_W \leftarrow (I-L_{W \times W})^{-1}\vy_W$; \hfill
\COMMENT{See \equationref{eq_step2_1}}
\STATE $\vl_Z \leftarrow \vy_Z + L_{Z \times W} \vl_W$; \hfill
\COMMENT{See \equationref{eq_step2_1}}
\STATE $\varepsilon_{min}=\min\{ \min_{i\in
  Z}\{\frac{0-[\vx+L\vq^t]_i}{\ell_i}\}, \min_{i \in
  W}\{\frac{1-[\vx+L\vq^t]_i}{\ell_i}\} \}$; \hfill \COMMENT{See
  \equationref{eq_step3}}
\STATE $p_{t+1} \leftarrow p_t - \varepsilon_{min}$;
\STATE $\uq(p)|_{[p_{t+1},p_t)} \leftarrow \langle\one_Z, \vr_W(p_t-p),
\one_O \rangle$; \label{code1:add}
\STATE $t \leftarrow t+1$;

\ENDWHILE
\STATE $\uq(p)|_{(-\infty,p_t)} \leftarrow \one$;
\RETURN {$\uq$};
\end{algorithmic}
\end{algorithm}




\medskip
{\bf \noindent \lemmaref{lem:pivot} (restated).}\quad
{\sl
For any $W_2\subset W$ s.t. $\rho(L_{W_2 \times W_2})\geq 1$, we have
$\forall
 \varepsilon > 0$,
 $[\uq(p_t-\varepsilon)]_k=1$.
}
\begin{proof}
We will only prove the statement when $W_2=W$, as the analysis for $W_2\neq W$ is similar.

  As $L_{W\times W}$ is a non-negative matrix, based on
  \lemmaref{lemma_matrix} there exists a non-zero eigenvector $\vu_W \geq
  \zero_W$ such that $L_{W\times W}\vu_W = \lambda\vu_W$ and
  $\lambda = \rho(L_{W\times W})\geq1$.  $\vu_W$ can be extended to $[n]$ by
  defining $\vu_{Z\cup O}=\zero_{Z\cup O}$.
  Let
  \begin{equation*}
    k=\arg\min_{k\in [n], u_k\neq0}{\frac{1-q_k^t}{u_k}}
  \end{equation*}

  Since $\vu \neq \zero$, the above equation is well defined.  As $q_i^t
  <1$ for any $i\in W$ in the current configuration, we also have
  $\frac{1-q_k^t}{u_k}>0$.  The tie is broken arbitrarily.

  Since $\vy >0$ and $\vu \ge 0$, for any $\varepsilon>0$, there exists
  $\delta>0$ satisfying $\delta \vu \leq \varepsilon \vy$.  Because
  $\delta$ can be arbitrary small, let $\delta=\left(
    \frac{1-q_k^t}{u_k} \right)/(1+\lambda+\cdots+\lambda^{m-1})$ in
  which $m$ is sufficiently large to satisfy the above constraint.
  For any probability vector $\vq$, we have
  $[f_{p_t-\varepsilon}(\vq)]_i = \med\{0,1,[\vx+\varepsilon\vy+L\vq]_i\}$.
  Define function $ [h(\vq)]_i =
  \med\{0,1,[\vx+\delta\vu+L\vq]_i\}$. Clearly,
  $f_{p_t-\varepsilon}(\vq) \geq h(\vq)$.

  Starting from $\vq^t \geq \vq^t$, we continue to apply the left side
  by $f_{p_t-\varepsilon}$ and the right side by $h$, we derive the
  followings: \footnote{ Within which we implicitly adopted the
    following property:
$ \forall m_0 < m, $
\[\delta(1+\lambda+...+\lambda^{m_0-1})\vu_W +
\vq_W^t \leq \left(\frac{1-q_k^t}{u_k}\right)\vu_W + \vq_W^t \leq
\one_W \]
}
\begin{small}
\begin{eqnarray*}
    f_{p_t-\varepsilon}(\vq^t) &\geq &h(\vq^t) \geq \langle \zero_Z,
    \vx_W'+\delta\vu_W + L_{W \times W}\vq_W, \one _O \rangle\\
    &=&\langle \zero_Z, \delta\vu_W + \vq_W^t,\one_O\rangle \\
    f_{p_t-\varepsilon}^{(2)}(\vq^t)&\geq&h(\langle \zero_Z, \delta\vu_W+\vq_W^t,\one_O\rangle)\\
    &\geq&\langle\zero_Z, \delta(I+L_{W\times{}W})\vu_W + \vq_W^t, \one_O \rangle\\
    &=&\langle(\zero_Z, \delta(1+\lambda)\vu_W+\vq_W^t,\one_O\rangle\\
    &\ldots \\
    f_{p_t-\varepsilon}^{(m)}(\vq^t) &\geq &\langle\zero_Z,\delta(1+\lambda+...+\lambda^{m-1})\vu_W+\vq_W^t,\one_O\rangle\\
    &=&\langle \zero_Z, \left(\frac{1-q_k^t}{u_k}\right)\vu_W + \vq_W^t,\one_O\rangle
\end{eqnarray*}
\end{small}

From our selection of $k$, we know that
\[\left[ \langle \zero_Z,\Big(\frac{1-q_k^t}{u_k}\Big)\vu_W +
  \vq_W^t, \one_O\rangle \right]_k = 1\]
i.e., $\forall \varepsilon>0$, we have $[\uq(p_t-\varepsilon)]_k \geq
[f_{p_t-\varepsilon}^{(m)}(\vq^t)]_k = 1$.  This completes the proof of the
existence of pivot $k$.
\end{proof}

 \medskip
 {\bf \noindent \lemmaref{lemma_alg2_new_1} (restated).}\quad
 {\sl
 Given the definition of $\vu$ in \equationref{eq_newu} and
$k$ using
 \equationref{eq_pick}, we have $\forall \varepsilon > 0,
 [\uq(p_t-\varepsilon)]_k=1.$
 }
 \begin{proof}
   We only prove the case when $Z=O=W\setminus W_2=\emptyset$, and will briefly describe how our proof can be
   extended to the general case.
   We use $\vq_{-w}$ to denote $\vq_{[n]\setminus \{w\}}=\vq_{W_1}$.

   Let $\delta = \min_{k\in [n], u_k\neq0}{\frac{1-q_k^t}{u_k}} > 0$. We
   know that if we increase from $\vq^t$ in the direction of $\vu$, we
   can at most raise $\delta\vu$ until agent $k$'s probability hits $1$.
   For a fixed $\varepsilon > 0$, let $\vq'=\uq(p_t-\varepsilon)$ be the
   pessimistic equilibrium.  To prove $q'_k=1$ we consider two cases:
 \begin{itemize}

 \item $q'_{w} \geq q^t_{w} + \delta$.

   This means that in the real scenario, agent $w$ indeed increases her
   probability by at least $\delta$.  It can be verified that in this
   case, the rest of the agents in $W_1=[n]\setminus\{w\}$ have to
   increase by at least $\delta\vu_{-w}$.  In other words, $\vq'-\vq^t
   \geq \delta\vu$ which already implies $q'_k \geq 1$ by our definition
   of $k$ and $\delta$.

 \item $q'_{w} < q^t_{w} + \delta$.

   In this case, the actual
   final probability of $w$ is small.   Let
   $\delta' = q'_w - q^t_w < \delta$.  We start from the inequality
   $\vq^t + \langle\zero_{-w},\delta'\rangle \geq \vq^t$.  Let ${\bf
     z}_{-w} = L_{W_1 \times \{w\}}$.  By applying the transfer
   function $f_{p_t-\varepsilon}$ to both sides and using the monotonicity,
   $$ \vq^t + \langle \varepsilon \vy_{-w} + \delta' {\bf z}_{-w}, \sigma_1
   \rangle = f_{p_t-\varepsilon}(\vq^t +
   \langle\zero_{W_1},\delta'\rangle) \geq f_{p_t-\varepsilon}(\vq^t) $$
   for some $\sigma_1 \geq 0$.  Based on $ q^t_w + \delta' \geq q'_w =
   [f_{p_t-\varepsilon}^\iter(\vq^t)]_w \geq [f_{p_t-\varepsilon}(\vq^t)]_w$,
   we always have $ \vq^t + \langle \varepsilon \vy_{-w} + \delta' {\bf
     z}_{-w}, \delta' \rangle \geq f_{p_t-\varepsilon}(\vq^t) $.  By
   applying the transfer function again we have
   $$ \vq^t + \langle \varepsilon (I+L_{W_1 \times W_1}) \vy_{-w} + \delta'
   (I+L_{W_1 \times W_1}) {\bf z}_{-w}, \sigma_2 \rangle \geq
   f_{p_t-\varepsilon}^{(2)}(\vq^t). $$ We continue to replace $\sigma_2$ by
   $\delta'$ and apply the transfer function.  Doing this iteratively
   while assuming that $\varepsilon$ is sufficiently small, we have:
   $$ \vq^t + \langle \varepsilon (I-L_{W_1 \times W_1})^{-1} \vy_{-w} +
   \delta' (I-L_{W_1 \times W_1})^{-1} {\bf z}_{-w}, \delta' \rangle \geq
   f_{p_t-\varepsilon}^\iter(\vq^t). $$ Recall the definition of $\vu$ we
   can rewrite the above equation as: $ \vq^t + \delta'\vu + \langle
   \varepsilon (I-L_{W_1 \times W_1})^{-1} \vy_{-w}, 0 \rangle \geq \vq'. $
   Since $\delta'<\delta$ and $\vq^t < \one$, we have $\vq^t + \delta'\vu
   < \one$.  When $\varepsilon$ is sufficiently small, we also have that the
   left hand side in the above equation is smaller than $\one$, and this
   proves that $\vq' < \one$ when $\varepsilon$ is small, which contradicts
   \lemmaref{lem:pivot} which says that the pivot always exists.

 \end{itemize}


 We describe how we prove the general case where $Z, W\setminus
 W_2$ and $O$ are not necessarily empty.  Imagine a subproblem with only
 $|W_2|$ rational players, while for agent $i \in [n]\setminus W_1$, her
 probability is fixed to $q^t_i$, no matter how the price varies and
 other players behave.  We can also define the transfer function and
 pessimistic equilibrium in this subproblem.  Then, using the same
 argument as above, we can find one pivot $k$ such that agent $k$'s
 probability hits $1$ in the subproblem, when $p<p_t$.  It can be
 verified that in the original problem, this agent $k$ will also buy with
 probability $1$, since when releasing the constraints on agents in $[n]
 \setminus W_2$, the entire probability vector may only increase rather
 than decrease.
 \end{proof}

\begin{algorithm}
\caption{\textsf{LineSweepMethod}($n,T,\mathbf{a},\mathbf{b}$)}
\label{alg_alg2}
\begin{algorithmic}[1]
\REQUIRE $n,T,\mathbf{a},\mathbf{b}$.
\ENSURE The pessimistic equilibrium function $\uq : p \mapsto \uq(p)$.

\STATE $L_{i,j} \leftarrow \inft_{j,i}/(b_i-a_i)$;
\STATE $p_1 \leftarrow \max_{1 \leq i \leq n}{b_i}$;
\STATE $\uq(p)|_{[p_1,\infty)} \leftarrow \zero$;
\STATE $Z \leftarrow [n]$; $W \leftarrow \emptyset$; $O \leftarrow \emptyset$; $t \leftarrow 1$;
\WHILE{$\uq(p_t) \neq \one$}

\STATE $\vq^t \leftarrow \uq(p_t)$;

\STATE $\vx \leftarrow ((b_1-p_t)/(b_1-a_1), (b_2-p_t)/(b_2-a_2), ...,
(b_n-p_t)/(b_n-a_n))$;
\STATE $\vy \leftarrow (1/(b_1-a_1), 1/(b_2-a_2), ..., 1/(b_n-a_n))^T$;

\FORALL{$i \in Z$ s.t. $\vx_i+\sum_j{L_{i,j}q_j^t}=0$}
\STATE $Z \leftarrow Z\setminus\{i\}$; $W \leftarrow W\cup\{i\}$;
\ENDFOR

\FORALL{$i \in W$ s.t. $\vx_i+\sum_j{L_{i,j}q_j^t}=1$}
\STATE $W \leftarrow W\setminus\{i\}$; $O \leftarrow O\cup\{i\}$;
\ENDFOR

\IF{$\rho(L_{W \times W})<1$}

\STATE $\vl_W \leftarrow (I-L_{W \times W})^{-1}\vy_W$;
and $\vl_Z \leftarrow \vy_Z + L_{Z \times W} \vl_W$; \hfill \COMMENT{See
  \equationref{eq_step2_1}}
\STATE $\varepsilon_{min}=\min\{ \min_{i\in Z}\{\frac{0-[\vx+L\vq^t]_i}{\ell_i}\}, \min_{i \in W}\{\frac{1-[\vx+L\vq^t]_i}{\ell_i}\} \}$; \hfill \COMMENT{See \equationref{eq_step3}}
\STATE $p_{t+1} \leftarrow p_t - \varepsilon_{min}$;
\STATE $\uq(p)|_{[p_{t+1},p_t)} \leftarrow \langle\one_Z, \vr_W(p_t-p), \one_O \rangle$; \label{code2:add}

\ELSE[\ensuremath{|W| \geq 2}]

\STATE Assume $W = \{w_1,w_2,...w_{|W|}\}$;

\FOR{$i\leftarrow 2$ to $|W|$}
\STATE $W_1 \leftarrow \{w_1,...w_{i-1}\}$; $W_2 \leftarrow \{w_1,...w_i\}$;
\IF{$\rho(L_{W_2\times W_2}) \geq 1$}
\STATE $\vu_{W_1} = (I-L_{W_1\times W_1})^{-1}L_{W_1 \times \{w_i\}}$; \\
$u_{w_i}=1$; $\vu_{[n]\setminus W_2}=\zero_{[n]\setminus W_2}$; \hfill \COMMENT{See \equationref{eq_newu}}
\STATE $k\leftarrow \mathop{\rm argmin}_{k \in [n], u_k \neq
 0}\{(1-[\vq^t]_k)/u_k\}$; \hfill \COMMENT{See \equationref{eq_pick}}
\STATE $O \leftarrow O \cup\{k\}$; $\bar O \leftarrow [n] \setminus O$;
\STATE $\forall i\in \bar O, \quad [a_i',b_i']=[ a_i+\sum_{j\in O}{\inft_{j,i}}, b_i+\sum_{j \in O}{\inft_{j,i}} ]$; \hfill \COMMENT{See \equationref{eq_alg2_1}}
\STATE $\uq' \leftarrow \text{\textsf{LineSweepMethod}}(|\bar O|, T_{\bar O \times \bar O}, \mathbf{a'}, \mathbf{b'})$;
\STATE $\uq(p)|_{(-\infty,p_t)} \leftarrow \langle\uq'(p), \one_O \rangle$; \label{code2:assign}
\RETURN{$\uq$};
\ENDIF

\ENDFOR
\hfill\COMMENT{never reach here}

\ENDIF

\STATE $t \leftarrow t+1$;

\ENDWHILE
\STATE $\uq(p)|_{(-\infty,p_t)} \leftarrow \one$;
\RETURN {$\uq$};
\end{algorithmic}
\end{algorithm}

\newpage

\medskip
{\bf \noindent \lemmaref{lemma_alg2_2} (restated).}\quad
{\sl
Let $\uq'(p)$ be the pessimistic equilibrium function in the
subproblem. We have:
$$\forall p<p_t, \uq(p)=\langle\uq'(p),\one_{O'}\rangle.$$
}
\begin{proof}
  We prove the lemma in two steps.  We will first show that
  $\langle\uq'(p),\one_{O'}\rangle$ is an equilibrium at price $p$, and
  then lower bound the pessimistic equilibrium by
  $\langle\uq'(p),\one_{O'}\rangle \leq \uq(p)$.  Combined with the
  property of equilibrium in \lemmaref{lemma_eqproperty}a, it is enough
  to see that $\langle\uq'(p),\one_{O'}\rangle$ is the pessimistic
  equilibrium of the original problem.

  For convenience let $\overline{O}'=[n] \setminus O'$.

\begin{itemize}
\item
Let $\vq=\langle\uq'(p),\one_{O'}\rangle$, and we are going to show
$f_p(\vq)=\vq$.
Based on the definition of $[a_i',b_i']$ in the subproblem, we already have that $[f_p(\vq)]_{\overline{O}'}=\vq_{\overline{O}'}=\uq'(p)$.
This is because $\forall i \in \overline{O}'$,
\begin{equation*}\begin{split}
[f_p(\vq)]_i &= \med\left\{0,1,\frac{b_i-p+\sum_{j\in [n]}{\inft_{j,i}q_j}}{b_i-a_i}\right\} \\
                        &= \med\left\{0,1,\frac{b_i'-p+\sum_{j\in \overline{O}'}{\inft_{j,i}q_j}}{b_i'-a_i'}\right\} = q_i.
\end{split}\end{equation*}
Therefore we only need to show that $[f_p(\vq)]_{O'}=\one_{O'}$.
Assume the contradiction that $[f_p(\vq)]_{O'} \leq \one_{O'}$ and $\exists i \in O'$ s.t. $[f_p(\vq)]_i<1$.
We start from $f_p(\vq) \leq \vq$ and arrive at $f_p^{(m)}(\vq) \leq f_p^{(m-1)}(\vq)$ by using the monotonicity of $f$.
The following limit exists because a non-increasing and lower bounded sequence has a limit.
$$\vq^* = \lim_{m \rightarrow \infty}{f_p^{(m)}(\vq)} \leq f_p(\vq) $$
Because of the continuity of function $f$, $\vq^*$ is an equilibrium at price $p$.
According to \lemmaref{lemma_eqproperty} $$[\uq(p)]_i \leq q_i^* \leq [f_p(\vq)]_i < 1.$$
If $i \in O = O' \setminus\{k\}$, this contradict the fact that $1=[\uq(p_t)]_i \leq [\uq(p)]_i$;
if $i=k$ this contradicts \lemmaref{lemma_alg2_new_1}.
Therefore it must be the case that $f_p(\vq)=\vq$.

\item
We now lower bound the pessimistic equilibrium $\uq(p)$.  For similar
reason as the first half of the proof, we have
$[\uq(p)]_{O'}=\one_{O'}$.  Let $f_p'$ be the transfer function of the
subproblem.  We start from the inequality $\langle
\zero_{\overline{O}'}, \one_{O'} \rangle \leq \uq(p)$ and apply the
monotone function $f_p$ to both sides:
$$f_p(\langle \zero_{\overline{O}'}, \one_{O'} \rangle) = \langle f_p'(\zero_{\overline{O}'}), \star \rangle \leq \uq(p)$$
We need not to know what $\star$ is, and start with the new inequality
$\langle f_p'(\zero_{\overline{O}'}), \one_{O'} \rangle \leq \uq(p)$
and derive that:
$$f_p(\langle f_p'(\zero_{\overline{O}'}), \one_{O'} \rangle) = \langle f_p'^{(2)}(\zero_{\overline{O}'}), \star \rangle \leq \uq(p)$$
By doing this again and again, we reach the inequality
\[\langle{}f_p'^\iter(\zero_{\overline{O}'}), \one_{O'} \rangle \leq \uq(p)\]
which immediately gives us $\langle\uq'(p),\one_{O'}\rangle \leq
\uq(p)$.
\end{itemize}

This completes the proof.
\end{proof}

\section{Missing Proofs in \sectionref{SEC:OTHER}}

\subsection{Hardness results with negative influences}
\label{SEC:HARDNESS}


In this section, we show that when the influence values can be
negative, it is PPAD-hard to compute an {\em approximate}
equilibrium. We define a probability vector $\vq$ to
be an $\varepsilon$-approximate
equilibrium for price $p$ if:
$$q_i \in (q_i'-\varepsilon, q_i'+\varepsilon),$$
where $q_i'=\med\{0,1,
\frac{b_i-p+\sum_{j\in[n]}T_{j,i}q_j}{b_i-a_i}\}$.

We prove the PPAD hardness by a reduction from the two player Nash
equilibrium computation.  Our construction is inspired
by~\cite{Chen2011}.  Let matrices $A, B \in {\mathcal{R}}^{n\times n}$ be the
payoff matrices of the two players respectively, i.e. $(A_i)_j$
(resp. $(B_i)_j$) is the payoff for the first player (resp. the second
player) when the first player plays its $i$-th strategy and the second
player plays its $j$-th strategy.  It is PPAD-hard to approximate the
two player Nash Equilibrium with error $1/n^\alpha$ for any constant
$\alpha>0$~\cite{Chen2009a}. We build an instance of our pricing problem as
follows. ($\delta$ is a small value to be determined later.)


\begin{itemize}
\item Price $p = 1/2$.
\item User $X_i$ with value interval $[0,1]$ for $i\in[n]$. The
  probability that $X_i$ buys the product is $x_i$.
\item User $Y_i$ with value interval $[0,1]$ for $i \in [n]$. The
  probability that $Y_i$ buys the product is $y_i$.
\item User $ U_{i,j}$, $i,j\in [n]$ is used to enforce $x_i =0$, when
  $A_iy^T+\delta < A_jy^T$. For any $k\in [n]$, we assign influence on edge
  $(Y_k, U_{i,j})$ to be $  (A_j)_k-(A_i)_k$. Define $U_{i,j}$'s valuation interval to be $[1/2-\delta,
  1/2-\delta+\delta^2]$.
\item User $ V_{i,j}$, $i,j\in [n]$ is used to enforce $y_i = 0$ when $B_ix^T
  +\delta < B_jx^T$. For any $k\in [n]$,
  influence on edge $(X_k, V_{i,j})$ is $(B_j)_k -
  (B_i)_k$. Define $V_{i,j}$'s valuation to be $[1/2-\delta,
  1/2-\delta+\delta^2]$.
\item For $i,j\in [n]$, influence values on edges $(U_{i,j}, X_i)$ and $(V_{i,j}, Y_i)$
  are $-1$.
\item All other pair-wise influence values are zero.
\end{itemize}

In our setting, if $U_{i,j}$ buys the product, it will provide
influence of $-1$ to $X_i$, which will imply the probability that
$X_i$ will buy the product is $0$.

\medskip
{\bf \noindent \theoremref{thm:hard} (restated).}\quad
{\sl
It is PPAD-hard to compute an $n^{-c}$-approximate equilibrium of our
pricing system for any $c>1$ when influences can be negative.
}
\begin{proof}
  Let $\delta = n^{-c}$. Consider the instance we constructed above.
  Let $\vx,\, \vy,\, \vu, \, \vv$ be the set of vectors that form an
  $\delta$-approximate equilibrium of our pricing instance.  We will
  show that we can construct an $O(n^{1-c})$ approximate Nash
  equilibrium for the two player game.
  To simply
  the notation, we define $x\pm y = [x-y, x+y]$.  In particular, we
  have

\begin{align*}
x_i &\in \mbox{med}\{0, 1, 1/2 - \sum\nolimits_{i} u_{i,j}\} \pm \delta\\
y_i &\in \mbox{med}\{0, 1, 1/2 - \sum\nolimits_{i} v_{i,j}\} \pm \delta\\
u_{i,j} &\in \mbox{med}\{0, 1, 1-1/\delta+1/\delta^2\langle
A_{j}-A_i, \vy\rangle\}\pm \delta\\
v_{i,j} &\in \mbox{med}\{0, 1, 1-1/\delta+1/\delta^2\langle
B_{j}-B_{i}, \vy\rangle\}\pm \delta
\end{align*}

For the purpose of controlling normalization, we first prove that
$||\vx||_\infty = ||\vy||_\infty \in 1/2\pm \delta$. It is clear that
$||\vx||_\infty \leq 1/2+\delta$, since $X_i$ receives no positive influence
in our construction. Furthermore, for any vector $\vy$, let $t =
\arg \max_{i \in [n]}\{A_i\vy^T\}$. Then for each $U_{t,j}$, the sum of
influence is $(A_j-A_t)\vy^T \leq 0$.
As a result, $U_{t,j}$ will
never buy the product and give a negative influence to $X_t$, which
implies $||\vx||_\infty \geq x_t \geq 1/2-\delta$.  The proof of $||\vy||_\infty \in
1/2\pm \delta$ is similar. We can define
\[[\vx']_i = \left\{ \begin{array}{ll}
[\vx]_i & \textrm{ if $[\vx]_i > \delta$}\\
0 & \textrm{otherwise}
\end{array} \right.\]
Similarly, we obtain $\vy'$. We then normalize them to $\vx^* =
\frac{\vx'}{||\vx'||_1} $ and $\vy^* =
\frac{\vy'}{||\vy'||_1}$. It is sufficient to prove that
$\vx^*$ and $\vy^*$ form
an $9n\delta$-approximate Nash for the two player game. In particular,
we shall show
\[ \langle A_i, \vy^*\rangle  + 6n\delta < \langle A_j, \vy^*\rangle \Longrightarrow x^*_i = 0\]
\[ \langle B_i, \vx^*\rangle + 6n\delta< \langle B_j, \vx^*\rangle \Longrightarrow y^*_i = 0\]
When $A_i\vy^* + 6n\delta < A_j\vy^*$, clearly $\langle A_j-A_i,
\vy'\rangle > 6n\delta||\vy'||_1>3n\delta$ and $\langle A_j-A_i, \vy\rangle >
3n\delta - 2n\delta\geq n\delta$. (The entries in $A$ and $B$ are
within range $[-1,1]$.) Therefore, $u_{i,j}\geq 1-\delta$,
which implies $x_i \leq \delta$ and $x^*_i=0$ by our construction. The
proof for the statement of $\vy^*$ is symmetric.
\end{proof}

\theoremref{thm:hard} implies that computing an exact equilibrium in
our pricing system is PPAD-hard, when the price is given and the
influence could be negative.



\subsection{Discriminative pricing model} \label{SEC:MULTIPRICE}

In this section, we discuss the extension of our problem in the discriminative pricing
model, in which different agents may be offered with different prices
to the same good, and there are at most $k$ different prices
offered. We only consider non-negative influences in this section.
Let $G$ be a $k$-partition of agent set $[n]$ and $g_i$ denote the group
which agent $i$ belongs to.  Let $\vp = (p_1, p_2,\ldots, p_k)$ be the
price vector corresponding to the $k$ groups in the partition.
\definitionref{def_probeq}, \definitionref{def:transfer}, and \definitionref{def:poeq}
for a single price $p$ can be straightforwardly extended to the case
of price vector $\vp$ with partition $G$, and we omit their
re-definitions here.  We define the revenue maximization problem under
the discriminative pricing model as follows.
\begin{definition}
The revenue maximization problem
is to compute an optimal price vector $\vp = (p_1, p_2,\ldots, p_k)$
w.r.t. the pessimistic equilibrium (resp. optimistic equilibrium):
\begin{small}
$$\argmax_{\vp\geq \mathbf{0}} \sum_{i\in [n]} p_{g_i}\cdot [\uq(\vp)]_i
\mbox{\; (resp. }
\argmax_{\vp\geq \mathbf{0}} \sum_{i\in [n]} p_{g_i}\cdot [\oq(\vp)]_i
\mbox{ )}.$$
\end{small}
\end{definition}

Apparently, the {\em uniform pricing} case is a special case of
discriminative model when $k=1$. In this section, we discuss two
different cases in this model: the fixed partition case and the
choosing partition case. As the name suggests, in the fixed partition
case, the partition of the agents are given. On the other hand, in the
choosing partition case, the algorithm has the flexibility to choose
the partition.


\subsubsection{Fixed partition case with constant $k$.}
\label{sec:givenpart}
In this case, we let the $k$-partition of $G$ be fixed and known to
our algorithm.  This is natural in a modern market such as setting prices
based on different regions or different user memberships.

Our algorithm for the uniform pricing model can be extended to some
restricted cases in the fixed partition case.  For instance, given a
fixed price vector $\vp =\langle p_1,p_2,\ldots,p_n\rangle$, we
consider (a) all possible price vectors that are $\{ \vp+x\one\,|\,
x\in \mathbb{R}\}$; or (b) all possible price vectors that are $\{\xi
\vp\,|\, \xi >0\}$. These two cases capture certain scenarios in which
the prices in different partitions either follow fixed ratios, e.g. by
different tax ratio or income distribution, or have fixed differences,
e.g. by transportation costs.  In both cases, we can reduce the
problem to a uniform price one, which can be solved by our proposed
algorithm. We only present the algorithm for the first case and the
proof for the second case is similar.

\begin{claim}
There is a refinement of \algorithmref{alg_alg2} for all possible price
vectors that are $\{ \vp+x\one\,|\, x\in \mathbb{R}\}$.
\end{claim}
\begin{proof}
In order to compute revenue with respect to price vector $\vp$,
we refine our line sweep method. Let $p_t = \min_{i \in [k]}{p_i}$
be the minimum entry in $\vp$ and $\Delta_i$ be $p_{g_i} - p_t$. We
use $\uq$ to denote the equilibrium when agents $i$ is offered price
$p+\Delta_i$ and modify \algorithmref{alg_alg2} line $7$ to
\[\vx\leftarrow \left(\frac{b_1-p_t-\Delta_1}{b_1-a_1},
  \frac{b_2-p_t-\Delta_2}{b_2-a_2}, ..., \frac{b_n-p_t-\Delta_n}{b_n-a_n}\right)\]
\end{proof}


If the space expanded by the price
vectors is not one dimensional, enumerating all structures like our
proposed line sweep algorithm is generally impractical. (See an
counter example in \appendixref{app:multiprice}.)


When there is no constrain on the possible prices, we design an FPTAS when
 $k$ is a constant.
We first estimate the optimal revenue which we can hope to achieve. In
particular, for any group $i\in[k]$, we set the prices for all other
groups to be $0$. By our algorithm in \sectionref{SEC:BAYESIAN}, we
can compute the maximum revenue from group $i$ in this case as
$R_i$. Clearly, the optimal revenue is at most $R = \sum_{i\in [k]}
R_k$. We then design a discretization scheme based on $R$.

Let $\varepsilon \in (0,1)$ be a constant. Define $p_{max} = R$ and
$p_{min} = \varepsilon R/(2kn)$.
Our algorithm works as follows:
\begin{itemize}
\item[1] Compute revenue $r_i$ when price vector is
    \[\vp_i=\left(0,
      (1+\varepsilon)^{i_1}p_{min},(1+\varepsilon)^{i_2}p_{min},\ldots, (1+\varepsilon)^{i_k}p_{min}\right)\]
for all $0\leq{}i_1,i_2,\ldots, i_k = \lceil \log_{1+\varepsilon}{2kn/\varepsilon}\rceil$.
\item[2] Return $\vp_i$ with the maximum calculated $r_i$.
\end{itemize}




\medskip
{\bf \noindent \theoremref{thm:fptas} (restated).}\quad
{\sl
There is an FPTAS for the discriminative pricing problem in the fixed
partition case with constant $k$.
}
\begin{proof}
The set of total prices for each group in the algorithm is
$O(\log_{1+\varepsilon}(n/\varepsilon)) = O(\frac{\log
  (kn/\varepsilon)}{\varepsilon})$. Enumerating all possible prices takes
time $O( \log^k(n/\varepsilon)/\varepsilon^k)$, which is polynomial when $k$
is constant.

Assume $\vp_{opt}$ be the optimal price vector with optimal revenue
$R_{opt} \geq \max_iR_i \geq R/k$. Let $\vp'$ be the price vector, which is obtained by
rounding all prices $\vp_{opt}$ down to the closest Steiner price. Now
consider the error introduced by the rounding scheme. Notice that by
monotonicity, this rounding will only increase  the buying
probability of each user. For all prices that are rounded to $0$, the
revenue from the users offered with those prices is at most
$\varepsilon R/(2k)\leq \varepsilon R_{opt}/2$ with the optimal price vector. All other prices are
decreased by at most a factor of
$1+\varepsilon/2$. The revenue collected from the agents offered with
those prices in $\vp'$ is at least $1+\varepsilon/2$ of that with
$\vp$. Therefore, in total, we receive a revenue of at least
$(1-\varepsilon/2)R_{opt}/(1+\varepsilon/2) \geq (1-\varepsilon)R_{opt}$.
%
%
\end{proof}

\subsubsection{Choosing partition case with constant $k$.}
\label{sec:choosepart}

Now we consider the case that the partition $G$ can be chosen by our
algorithm in order to maximize the seller's revenue. More precisely we
define our problem as follows. Given the distribution of agents'
values and their influence network, the problem is to compute the
optimal $k$-partition of $G$ together with an optimal price vector
$\vp$ to maximize the seller's revenue.  We prove that when the
revenue is measured based on the
pessimistic equilibrium, this optimization problem is NP-hard even in
the fixed valuation case ($a_i=b_i$ for each player $i$).

In particular, we consider the following
special case of the problem: (i) k = 2, (ii) The valuation of the
agents is deterministic, and (iii) the price can only be 0 or 1. For
the case $k > 2$, we can add some dummy agents in our construction and
force the optimal solution to get the optimal revenue in our
construction for $k = 2$. We summarize the main result in the
following theorem.

\medskip
{\bf \noindent \theoremref{thm:prehard} (restated).}\quad
{\sl
It is NP-hard to compute the optimal pessimistic discriminative
pricing equilibrium in the choosing partition case.
}
\begin{proof}
We use a reduction from the \emph{Vertex Cover} problem. We show that
using any polynomial algorithm for the pessimistic discriminative
pricing problem in
choose partition case, any instance of the \emph{Vertex Cover} problem
can be solved in polynomial time. In an instance of an \emph{Vertex
  Cover} problem, given a graph $G=(V,E)$, we must specify whether a
subset $S \subset V$ exists such that $|S| \leq K$ and $\forall u,v$ such
that $(u, v) \in E$, we have $v \in S$ or $u \in S$.

Then we prove that, for each graph $G(V, E)$, the exists a network
$G'(V', E')$ and agents' valuation so that, if $r_{opt}$ is the
optimal revenue in $G'$, vertex number in minimum vertex cover of $G$
is $|V'| - r_{opt}$. First we will show how to construct $G'$ from
$G$. $V'$ is formed from the union of three parts, denoted by $A$, $D$
and $M$. In the first set $A$, there is one vertex $a_i$ with initial
value 0 for each vertex $v_i$ in $G$. The set $D$ is used to represent
the edges in original graph $G$. There is a vertex $d_e$ for each edge
in $G$. The initial values of all these vertices are 0. Let $e = (v_i,
v_j) $ be an edge in $G$. There is one edge from $a_i$, one from $a_j$
to $d_e$ weighted 1. The edges is used to represent the cover action,
if $a_i$ or $a_j$ buys the product, the $d_e$ vertex will also buy the
product. In addition, we also need construct $|D| \times |A|$ edges
weighted $\frac{1}{|D|}$ in $G'$, which from each $d_e$ to each
$a_i$. The edges means only if we cover all the vertex in $D$, all the
vertex in $A$ will just reach the value 1. Finally, we must use a
considerable large set $M(\geq |V|^3)$ to force the optimal solution
to set an price on 1. This is because if no final price is 1, it is
hard to guarantee all the vertex in $D$ will buy the product, which
also represent all the edges in $G$ be covered. Therefore, we put
independent vertexes in $M$ and weight them with 1. It is obvious that
we must set a price 1 and another 0 to get the best revenue and
activate the vertexes in $A$ and $D$. \par

Now we will show the optimal revenue $r_{opt}$ in $G'$ is equal to
$|V'| - |C|$. Let $C$ be minimum vertex cover of $G$ and $F$ be the
set of people who get a zero price in $G'$. In our configuration, a
final price of optimal solution must be 1, so the revenue is equal to
$|V' - F|$. Firstly, we show that $r_{opt} \geq |V| - |C|$.  If given
a minimum vertex cover $C$ for $G$, we can define $F$ according to
$C$. It means the seller will give the free product to the vertex in
$A$ if it represent a vertex in minimum cover $C$. By the definition
of vertex covering, all the vertexes in $D$ , which represents the all
edges in original graph $G$,  will be activated and their value will
all reach to 1. As a result, the value of all vertexes in $A$ will
also reach to 1 and will buy the product. Conclusively, our revenue
will reach to $|V'| - |C|$. At last, we will prove $r_{opt} \leq |V'|
- |C|$. Suppose $r_{opt} > |V'| - |C|$, there must be a free set $F$
to achieve the maximum revenue $r_{opt}$. Consider the structure of
$F$, if $F \cap M$ is not empty, we can eliminate these vertexes to
get a better revenue. If $F \cap D$ is not empty, each point $d \in F
\cap D$ can be replaced by the vertex in $A$ which have an edge to
it. This replacement never decreases our revenue because the new
vertex have a 1-weight edge to the old vertex. So there must be a $F
\subseteq A$, which could make the revenue greater than $|V'| -
|C|$. By the construction, we can convert $F$ to a vertex cover in
$G$. This is a contradiction to the definition of minimum vertex
cover.
\end{proof}

\section{Counter Example in \appendixref{SEC:MULTIPRICE}}\label{app:multiprice}

Assume $n$ is even.  Let $p_1$ be the price offered to agents
$\{1,3,\ldots, n-1\}$ and $p_2$ be the price offered to
$\{2,4,\ldots,n\}$. The influences are defined as $\inft_{j,i} =
2^{\lceil j/2-1\rceil}$ for $ i < j$ and $j-i$ is odd and greater than
$0$, and $0$ otherwise.  The valuation of agent $i$ is $2^{\lceil
i/2-1\rceil}$.  There are a total of $2^{\Omega(n)}$ structures for
the pessimistic equilibrium as $p_1$ and $p_2$ vary in $[0, +\infty)$.

\begin{proof}
We prove the following stronger statement by induction: for all prices
$p_1 \in (0, 2^{n/2})$ and $p_2\in (0,2^{n/2})$, there are at least $2^{n/2}$
structures.

Consider the base
case of $n=2$, with
agent 1 and 2. Since there is no influence among them, the number of
structure configuration is certainly $4 > 2$, with the price range
$p_1\in (0,2), \, p_2\in(0,2)$. Suppose the statement
is true for $n=2i$. For the case of $n=2(i+1)$, there are two additional
agents $2i+1$ and $2i+2$.

Consider the price range $p_1 \in (2^{i},2^{i+1}),\, p_2\in
(0,2^{i})$.  Agent
$2i+1$ will not buy the product while agent $2i+2$ does in this case. Since the
influence from agent $2i+2$ to every agent  with price $p_1$ is
$2^{i}$ except $2i+1$,
the ``effective price'' for all odd agents except $2i+1$ is $p_1 - 2^i \in
(0,2^{i})$. In such price range, there are at least $2^i$ structures
by induction. Symmetrically, the same conclusion holds
for price range $p_1 \in (0, 2^{i}),\, p_2\in (2^i,2^{i+1})$. Notice
these two price ranges have difference configuration for agents $2i+1$
and $2i+2$.  Therefore, in total there are at least $2\cdot 2^i
=2^{i+1}$ structures.
\end{proof}

\clearpage
\bibliographystyle{alpha}
\bibliography{game}

\end{document}